\documentclass[openacc]{rstransa}

\newtheorem{theorem}{\bf Theorem}[section]
\newtheorem{lemma}[theorem]{\bf Lemma}
\newtheorem{definition}[theorem]{\bf Definition}
\newtheorem{remark}[theorem]{\bf Remark}

\newtheorem{corollary}[theorem]{\bf Corollary}

\newcommand{\bomega}{\mbox{\boldmath{$\omega$}}}

\newcommand{\etab}{\mbox{\boldmath{$\eta$}}}

\newcommand{\eref}[1]{{(\ref{#1})}}
\newcommand{\mbf}[1]{{\mathbf{#1}}}
\newcommand{\Lie}[1]{\mathop{\mathcal{L}}_{\scriptsize {#1}}}

\jname{rsta}
\Journal{Phil. Trans. R. Soc}

\begin{document}

\title{On the role of continuous symmetries in the solution of the 3D Euler fluid equations and related models}

\author{
Miguel D. Bustamante}

\address{School of Mathematics and Statistics, University College Dublin, Belfield, Dublin 4, Ireland}

\subject{Fluid Mechanics, Differential Equations, Applied Mathematics}


\corres{Miguel D. Bustamante\\
\email{miguel.bustamante@ucd.ie}}

\begin{abstract}
We review and apply the continuous symmetry approach to find the solution of the 3D Euler fluid equations in several instances of interest, via the construction of constants of motion and infinitesimal symmetries, without recourse to Noether's theorem. We show that the vorticity field is a symmetry of the flow, so if the flow admits another symmetry then a Lie algebra of new symmetries can be constructed. For steady Euler flows this leads directly to the distinction of (non-)Beltrami flows: an example is given where the topology of the spatial manifold determines whether extra symmetries can be constructed. Next, we study the stagnation-point-type exact solution of the 3D Euler fluid equations introduced by Gibbon \emph{et al.} (\emph{Physica D}, vol.~132, 1999, pp.~497--510) along with a one-parameter generalisation of it introduced by Mulungye \emph{et al.} (\emph{J. Fluid Mech.}, vol.~771, 2015, pp.~468--502). Applying the symmetry approach to these models allows for the explicit integration of the fields along pathlines, revealing a fine structure of blowup for the vorticity, its stretching rate, and the back-to-labels map, depending on the value of the free parameter and on the initial conditions. Finally, we produce explicit blowup exponents and prefactors for a generic type of initial conditions.
\end{abstract}

\begin{fmtext}
\section{Introduction}
\label{sec:Intro}
Consider a smooth (i.e., of class $C^\infty$) 3D vector field 
\begin{eqnarray}
\nonumber
\mbf{u}: \Omega \times [0,T) \!\!\!\!&\longrightarrow & \!\!\!\!\mathbb{R}^3\\
\nonumber
(\mbf{x},t) \!\!\!\!&\longmapsto & \!\!\!\!\mbf{u}(\mbf{x},t) := (u^1(\mbf{x},t), u^2(\mbf{x},t), u^3(\mbf{x},t)),
\end{eqnarray}   
where $T>0$ and $\mbf{x}:=(x^1,x^2,x^3)$ denote local Cartesian coordinates of a point in the spatial domain $\Omega.$ It is customary to use superscripts for the components of coordinates \, and \, vector\, \, fields, \, as\, small \,displacements \break
\end{fmtext}
%

\maketitle
\noindent transform as components of contravariant vectors under general coordinate transformations. In this work we will set either $\Omega=\mathbb{T}^3$, or $\Omega=\mathbb{R}^2 \times \mathbb{T}$,  or $\Omega=\mathbb{T}^2 \times \mathbb{R}$, where $\mathbb{T}:=[0,2\pi)$ is the torus. In what follows, we will assume for simplicity that $\mbf{u}$ is divergence-free:
\begin{equation}
\nonumber
\nabla \cdot \mbf{u} := \frac{\partial u^1}{\partial x^1}+ \frac{\partial u^2}{\partial x^2} + \frac{\partial  u^3}{\partial x^3} = 0\,, \quad \text{for \,\, all} \quad \mbf{x} \in \Omega, \quad t \in [0,T).
\end{equation} 
The given vector field $\mbf{u}$ defines a 3D flow in $\Omega$ by the following non-linear, non-autonomous system of ordinary differential equations for the variables $\mbf{X}(t):=(X^1(t),X^2(t),X^3(t))$:
\begin{equation}
\label{eq:ODE}
 \dot{X}^1 = u^1(\mbf{X},t)\,, \qquad \quad
 \dot{X}^2 = u^2(\mbf{X},t)\,, \qquad \quad
 \dot{X}^3 = u^3(\mbf{X},t)\,, \qquad \quad t \in [0,T)\,.
\end{equation}
Solutions to this system subject to initial conditions $X^1(0)=X_0, X^2(0)=Y_0, X^3(0)=Z_0$ are known as pathlines, or characteristics, of $\mbf{u}$.  As $\mbf{u}$ is of class $C^\infty$ for $t \in [0,T)$, it is then Lipschitz, so pathlines with different initial conditions do not cross, but this may be violated if at some time $t\geq T$ the velocity field becomes non-regular.

In physical terms one could interpret the field $\mbf{u}$ as a velocity field for some incompressible fluid. The pathlines would then be the trajectories of massless particles.

In this paper we will apply infinitesimal symmetry methods, valid for an arbitrary smooth vector field $\mbf{u}$, for solving system \eqref{eq:ODE}, in order to establish non-trivial results when $\mbf{u}$ is a fluid velocity field satisfying the 3D Euler fluid equations or related models.

Section \ref{sec:2} is devoted to a review of the infinitesimal symmetry approach for general 3D vector fields, and the construction of constants of motion starting from symmetries, without using Noether's theorem. Section \ref{sec:3} is devoted to the application of this general work to the solution of the 3D Euler fluid equations. In it, we will look at general results, then provide an application to steady flows and finally an application to the stagnation-point like solution introduced in \cite{Gibbon1999497}. Section \ref{sec:4} is devoted to the main application, namely the $1$-parameter family of models introduced in \cite{mulungye2015lambda_m3/2} that generalise the case of \cite{Gibbon1999497}, containing it as a particular case. We construct constants of  motion and use them to solve for the main quantities of interest along pathlines, establishing a formula for the blowup time, and conditions on the blowup of the main fields, depending on the initial data and on the value of the model's parameter. Finally, conclusions are given in section \ref{sec:5}.

\section{The continuous symmetry approach}
\label{sec:2}
In this section we will present a simple exposition of the main tools of interest in the solution of the ODE system \eqref{eq:ODE}. These tools are ultimately based on the formalism of tensor bundles, where tensorial objects of various ranks are transported by the vector field $\mbf{u}$. We will avoid the use of this formalism in this work, as it tends to obscure the exposition of the applications. However, we will introduce (at a formal level) the main concept pertaining to the fundamental approach: the Lie derivative, which is extremely useful when proving the most relevant results.

\subsection{Basic tools: Constants of motion and Symmetries}

The most basic tool for solving the ODE system \eqref{eq:ODE} is the constant of motion. 

\begin{definition}[Constant of motion]
\label{def:CoM}
A smooth scalar function of position and time \mbox{$C: \Omega \times [0,T) \rightarrow \mathbb{R}$}
is a constant of motion for system \eref{eq:ODE} if it satisfies the linear partial differential equation
\begin{equation}
\label{eq:CoM}
\frac{\partial C}{\partial t} + (\mbf{u} \cdot \nabla) C = 0\,, \quad \text{for \,\, all} \quad \mbf{x} \in \Omega, \quad t \in [0,T),
\end{equation}
where $\mbf{u}\cdot \nabla := u^1\frac{\partial}{\partial x^1} + u^2\frac{\partial}{\partial x^2} + u^3\frac{\partial}{\partial x^3}$. In words, a constant of motion takes a fixed numerical value along a given pathline of the vector field $\mbf{u}$.
\end{definition}

\begin{remark}
Notice that if one finds a constant of motion with a nontrivial dependence on the space and time variables, then in principle the number of independent variables of system \eqref{eq:ODE} can be reduced by one unit via the formula
$C(\mbf{X}(t),t) = C(\mbf{X}(0),0)$, where $\mbf{X}(t)$ is a pathline. By iterating this process, when enough functionally independent solutions to this equation exist, one can in principle solve the system \eqref{eq:ODE} explicitly. The difficulty now lies of course in solving the linear partial differential equation \eqref{eq:CoM}. 
\end{remark}

The next basic tool in terms of complexity is the infinitesimal symmetry. Let us first recall Sophus Lie's transformation theory \cite{lie1891vorlesungen}, restricting the discussion to the so-called contemporaneous case, namely when the time is not transformed \cite[Chart A.6, Definition 3]{santilli1982foundations}.

\begin{definition}[Finite contemporaneous symmetry]
\label{def:finite_symm}
An invertible mapping defined continuously for each $t \in [0,T)$ 
\begin{eqnarray}
\nonumber
f_t: \Omega &\longrightarrow &\Omega\,, \qquad t \in [0,T),\\
\nonumber
(x^1,x^2,x^3) &\longmapsto & (\widetilde{x}^1, \widetilde{x}^2, \widetilde{x}^3) := f_t(x^1,x^2,x^3)\,,
\end{eqnarray}   
is called a finite contemporaneous symmetry for system \eqref{eq:ODE} if, under the mapping $\mbf{\widetilde{X}}(t):=(\widetilde{X}^1(t),\widetilde{X}^2(t),\widetilde{X}^3(t)) = f_t(X^1(t), X^2(t), X^3(t))$, system \eqref{eq:ODE} maps to the  system in the same form,
\begin{equation}
\label{eq:ODE_mapped}
 \dot{\widetilde{X}}{}^1 = u^1(\mbf{\widetilde{X}},t)\,, \qquad \quad
 \dot{\widetilde{X}}{}^2 = u^2(\mbf{\widetilde{X}},t)\,, \qquad \quad
 \dot{\widetilde{X}}{}^3 = u^3(\mbf{\widetilde{X}},t)\,, \qquad \quad t \in [0,T)\,.
\end{equation}
It is crucial in this definition that the form be maintained: the exact same functions $u^1, u^2, u^3$ are used in \eqref{eq:ODE} and \eqref{eq:ODE_mapped}. In words, a finite contemporaneous symmetry maps pathlines of $\mbf{u}$ into pathlines of $\mbf{u}$.
\end{definition}

The latter word definition is key to the utility of finite symmetries for solving the ODE system \eqref{eq:ODE}: once a solution is known, one can use the finite symmetry to construct new solutions.

It is clear from this definition that the set of all finite symmetry transformations for system \eqref{eq:ODE} define a group $G$, whose product rule is induced by the composition operation of functions:
\begin{eqnarray}
\nonumber
 G \times G &\longrightarrow & G\\
\nonumber
(f,g) &\longmapsto & h   \qquad \text{such \,\,that} \qquad h_t := f_t \circ g_t\,, \qquad t \in [0,T)\,.
\end{eqnarray}   
Infinitesimal symmetry transformations are simply near-identity versions of finite symmetry transformations (up to first order), and correspondingly they form an algebra, known as Lie algebra.  While Lie introduced the main techniques to construct infinitesimal symmetry transformations, he did not focus on the problem of finding infinitesimal symmetries for an arbitrary vector field $\mbf{u}$, which is what we will do now. The main equations were first introduced in \cite{hojman1985s}.  A modern account is given by Hojman in \cite{hojman1996non} (see also \cite{olver2000applications}). Because of its importance in what follows, we re-derive the equations determining contemporaneous infinitesimal symmetries of system \eqref{eq:ODE}. For simplicity, from here on we will restrict the discussion to the subgroup of volume-preserving finite contemporaneous symmetries, namely mappings satisfying the additional constraint $\det \frac{\partial({\widetilde{X}}{}^1, {\widetilde{X}}{}^2, {\widetilde{X}}{}^3)}{\partial(X^1, X^2, X^3)} = 1$. In terms of infinitesimal contemporaneous symmetries defined below, this extra restriction translates into the divergence-free condition $\nabla \cdot \etab = 0$.

\begin{definition}[Infinitesimal contemporaneous symmetry] 
\label{def:symm}
Consider a one-parameter family of volume-preserving finite contemporaneous symmetries $f_{t, \epsilon}$ depending on the continuous parameter $\epsilon \in \mathbb{R}$, and such that $f_{t, \epsilon=0}$ is the identity transformation. Then, the smooth divergence-free vector field 
\begin{eqnarray}
\nonumber
\etab: \Omega \times [0,T) &\longrightarrow & \mathbb{R}^3\\
\label{eq:eta_defn}
(\mbf{x},t) &\longmapsto & \etab(\mbf{x},t) := (\eta^1(\mbf{x},t), \eta^2(\mbf{x},t), \eta^3(\mbf{x},t))\,,\\
\label{eq:eta_defn_incomp}
&&\nabla \cdot \etab  = 0 \quad  \text{for \,\, all} \quad \mbf{x} \in \Omega, \quad t \in [0,T),
\end{eqnarray} 
defined by
\begin{equation}
\label{eq:etab_from_X_eps}
\etab(\mbf{x},t) = \left.\frac{\partial  f_{t, \epsilon}(\mbf{x})}{\partial \epsilon}\right|_{\epsilon=0}\,,
\end{equation}
is called an infinitesimal contemporaneous symmetry for system \eqref{eq:ODE}.
\end{definition}

A geometrical interpretation follows directly from the above definition: the functions 
\begin{equation}
\label{eq:epsilon_mapping}
\mbf{\widetilde{X}}(t) = \mbf{X}(t) + \epsilon\, \etab(\mbf{X}(t),t)\,,
\end{equation}
where $\mbf{\widetilde{X}}(t)=(\widetilde{X}^1(t), \widetilde{X}^2(t), \widetilde{X}^3(t))$, satisfy \eqref{eq:ODE_mapped} up to and including terms of order ${\mathcal{O}}(\epsilon)$ whenever $\mbf{X}(t)=(X^1(t),X^2(t),X^3(t))$ satisfy \eqref{eq:ODE}.
In words, an infinitesimal contemporaneous symmetry maps pathlines of $\mbf{u}$ into nearby pathlines of $\mbf{u}$.

We enunciate 
 our main theorem regarding infinitesimal symmetries:

\begin{theorem}
\label{thm:symm}
A smooth divergence-free vector field $\etab$ defined as in equations \eqref{eq:eta_defn}--\eqref{eq:eta_defn_incomp} is an infinitesimal contemporaneous symmetry if, and only if, it satisfies the following linear partial differential equation:
\begin{equation}
\label{eq:symm}
\frac{\partial \etab}{\partial t} + (\mbf{u} \cdot \nabla) \etab - (\etab \cdot \nabla) \mbf{u} = \mbf{0}\,,\quad  \text{for \,\, all} \quad \mbf{x} \in \Omega, \quad t \in [0,T),
\end{equation}
where $\mbf{0}$ denotes the vector whose components are all equal to zero.
\end{theorem}
\begin{proof} The proof is a direct exercise and will be omitted.
\end{proof}

\subsection{Lie derivatives} 

Given a vector field $\mbf{u}$, the equations that determine constants of motion \eqref{eq:CoM} and infinitesimal contemporaneous symmetries \eqref{eq:symm} are 
best written in terms of the linear operator $\Lie{\mbf{u}}$, called the Lie derivative along the vector field $\mbf{u}$, which maps tensors into tensors of the same rank and satisfies the Leibniz product rule of differentiation. Without going into the details of the Lie derivative's natural definition in terms of transport on tensor bundles (see \cite{santilli1982foundations, nakahara2018geometry} for details, and  \cite{arnold2008topological, besse2017geometric, gilbert2018geometric} for applications in fluids) we simply define the Lie derivative in the two cases of interest:

\begin{definition}
Let $\mbf{V}, \mbf{W}: \Omega \times [0,T) \to \mathbb{R}^3$ be two smooth divergence-free vector fields and let $B: \Omega \times [0,T) \to \mathbb{R}$ be a smooth scalar function. The Lie derivative of $B$ along $\mbf{V}$ is defined as the smooth scalar function
\begin{equation}
\nonumber 
\Lie{\mbf{V}} B :=\mbf{V} \cdot \nabla B\,,\quad  \text{for \,\, all} \quad \mbf{x} \in \Omega, \quad t \in [0,T)\,.
\end{equation}
Similarly, the Lie derivative of $\mbf{W}$ along $\mbf{V}$ is defined as the smooth divergence-free vector field
\begin{equation}
\label{eq:Lie_vector}
\Lie{\mbf{V}} \mbf{W} := \mbf{V} \cdot \nabla  \mbf{W} - \mbf{W} \cdot \nabla  \mbf{V} \,,\quad  \text{for \,\, all} \quad \mbf{x} \in \Omega, \quad t \in [0,T)\,.
\end{equation}
We say  that two vector fields $\mbf{V}, \mbf{W}$ commute  if $\Lie{\mbf{V}} \mbf{W} = \mbf{0}$.
\end{definition}

\begin{remark}
\label{rem:Lie}
In the case of vector fields, it follows from \eqref{eq:Lie_vector} that the Lie derivative defines a Lie algebra of divergence-free vector fields via the bilinear binary operation known as the Lie bracket, $[\mbf{V},\mbf{W}] := \Lie{\mbf{V}} \mbf{W}$, satisfying
\begin{eqnarray}
\nonumber 
\Lie{\mbf{W}}\mbf{V} + \Lie{\mbf{V}} \mbf{W}  =  \mbf{0} &\qquad & \text{(anticommutativity)}\,,\\
\label{eq:Jacobi}
\Lie{\mbf{U}} \Lie{\mbf{V}} \mbf{W} + \Lie{\mbf{V}} \Lie{\mbf{W}} \mbf{U} + \Lie{\mbf{W}} \Lie{\mbf{U}} \mbf{V} =  \mbf{0} 
&\qquad & \text{(Jacobi \,\, identity)}\,,
\end{eqnarray}
where $\mbf{U},\mbf{V}, \mbf{W}$ are any divergence-free vector fields.
\end{remark}

As a result of these definitions,  the equations defining constants of motion and infinitesimal symmetries become extremely simple:
\begin{corollary}
Under the hypotheses of definitions \ref{def:CoM} and \ref{def:symm}, a smooth scalar $C$ is a constant of motion and a smooth divergence-free vector field $\etab$ is an infinitesimal contemporaneous symmetry if and only if
\begin{eqnarray}
\label{eq:CoM_Lie}
\frac{\partial C}{\partial t} + \Lie{\mbf{u}} C  &=& 0\,,\quad \text{for \,\, all}\quad \mbf{x} \in \Omega, \, t \in [0,T)\,,\\
\label{eq:symm_Lie}
\frac{\partial \etab}{\partial t} + \Lie{\mbf{u}} \etab &=& \mbf{0}\,,\quad \text{for \,\, all}\quad \mbf{x} \in \Omega, \, t \in [0,T)\,.
\end{eqnarray}
\end{corollary}
\begin{proof}
The proof follows directly from the respective definitions of Lie derivative.
\end{proof}

\subsection{The gradient of a constant of motion}
Taking the gradient of equation \eqref{eq:CoM} gives an equation for the gradient of a constant of motion:
\begin{equation}
\label{eq:gradCoM}
\frac{\partial \nabla C}{\partial t} + (\mbf{u}\cdot \nabla) \nabla C+ (\nabla \mbf{u})^T \nabla C  = \mbf{0}\,, \quad \text{for \,\, all}\quad \mbf{x} \in \Omega, \, t \in [0,T)\,,
\end{equation}
where $(\nabla \mbf{u})^T$ is the square matrix with components $[(\nabla \mbf{u})^T]_{i}^{\,\,j} := \frac{\partial u^j}{\partial x^i}$. One has the result:
\begin{lemma}
\label{lem:grad}
Suppose the gradient of a scalar function $C$ satisfies equation \eqref{eq:gradCoM}. Then the scalar function $C$ is a constant of motion for system \eqref{eq:ODE}, up to an additive function of time that can be found.
\end{lemma}
\begin{proof}
Retracing the steps leading to \eqref{eq:gradCoM} from  \eqref{eq:CoM} we find
\mbox{$\nabla\left(\frac{\partial C}{\partial t} + ({\mbf{u}}\cdot \nabla) C\right)=\mbf{0}$,} implying $\frac{\partial C}{\partial t} + ({\mbf{u}}\cdot \nabla) C = f(t)$ for some scalar function $f$. Thus, the scalar \mbox{${C}'(\mbf{x},t) := C(\mbf{x},t) - \int f(t) \mathrm{d}t$} clearly satisfies 
$\frac{\partial C'}{\partial t} + ({\mbf{u}}\cdot \nabla) C' = 0$, which means that ${C}'(\mbf{x},t)$ is a constant of  motion.
\end{proof}

Unsurprinsingly, equation \eqref{eq:gradCoM} is yet another instance of a transport equation involving Lie derivatives, this time for dual vector fields, of which the gradient is a particular case. Solutions of \eqref{eq:gradCoM} where $\nabla C$ is replaced by a dual vector that is not a gradient, are called Lagrangian $1$-forms \cite{Bus03a, Bus03}, and are used in the construction of action principles for system \eqref{eq:ODE}, in what is known as the inverse problem of the calculus of variations. See  \cite{hojman1985s, hojman1996non, olver2000applications} and references therein for details, including Lie-advected higher-order tensors such as $2$-forms (called symplectic forms), antisymmetric $(2,0)$ tensors (called Poisson structures) and $(1,1)$ tensors (called recursion operators). Also, notice that Cauchy invariants (and generalised versions thereof) correspond to Lie-transported $p$-forms that are exact, such as the gradient, which is an exact $1$-form. See \cite{Zakharov89, kuznetsov2008mixed, besse2017geometric} and references therein for details.

\subsection{Zeroes of an infinitesimal symmetry and regularity of the flow}

We have an important new result on the time evolution of the zeros of infinitesimal symmetries:

\begin{theorem}[Conservation of zeroes of an infinitesimal symmetry and regularity of the flow] 
\label{thm:conservation_of_zeroes}
 Let $\mbf{u}:\Omega\times [0,T)\to\mathbb{R}^3$ be a smooth vector field and let $\mbf{X}(t)$ be a solution of the corresponding system \eref{eq:ODE} for $t \in [0,T)$, with initial position $\mbf{X}(0) = (X_0,Y_0,Z_0).$ Let $\etab(\mbf{x},t)$ be   a smooth infinitesimal contemporaneous symmetry of system \eref{eq:ODE}. Then, for any $t\in [0,T)$ we have $\etab(\mbf{X}(t),t) = \mbf{0}$ if and only if  $\etab(X_0,Y_0,Z_0,0) = \mbf{0}.$
 
In words, if the governing flow is differentiable then the zeroes of any infinitesimal symmetry follow the pathlines, and no new zeroes can be generated at any time $t \in (0,T)$.
\end{theorem}

\begin{proof}
Consider the vector function $\mbf{f}:[0,T)\to\mathbb{R}^3$ defined by $\mbf{f}(t) := \etab(\mbf{X}(t),t)$. Using the chain rule, its total time derivative satisfies
$\dot{\mbf{f}}(t) = \frac{\partial \etab}{\partial t}(\mbf{X}(t),t) + (\dot{\mbf{X}}(t)\cdot \nabla) \etab(\mbf{X}(t),t)$. By virtue of \eqref{eq:ODE} we have $\dot{\mbf{X}}(t)=\mbf{u}(\mbf{X}(t),t)$ and thus $\dot{\mbf{f}}(t) = \frac{\partial \etab}{\partial t}(\mbf{X}(t),t) + (\mbf{u}(\mbf{X}(t),t)\cdot \nabla) \etab(\mbf{X}(t),t)$. As $\etab$ is an infinitesimal symmetry, from \eqref{eq:symm} we get $\dot{\mbf{f}}(t) = (\etab(\mbf{X}(t),t)\cdot \nabla) \mbf{u}(\mbf{X}(t),t) = (\mbf{f}(t)\cdot \nabla) \mbf{u}(\mbf{X}(t),t)$. The latter is a non-autonomous linear homogeneous first-order differential equation for the vector $\mbf{f}(t)$. If the vector field $\mbf{u}$ is smooth then so are the matrix components of this equation, $\nabla \mbf{u}(\mbf{X}(t),t)$. From uniqueness of solutions to first-order ODEs we get $\mbf{f}(t)=\mbf{0}$ if and only if $\mbf{f}(0)=\mbf{0}$, which gives the desired result.
\end{proof}

This theorem is quite relevant to the scenario of finite-time singularity in fluid equations when we interpret $\mbf{u}$ as the velocity field. It shows that zeroes of infinitesimal symmetries can only be created when a solution $\mbf{u}$ becomes non-regular (see remark \ref{rem:blowup_and_zeroes} for an application). A similar result can be obtained for the gradient of a constant of motion $C$:  the zeroes of $\nabla C$ follow the pathlines, and no new zeroes can be generated at any time $t\in (0,T)$.

\subsection{Constructing new symmmetries and constants of motion from known symmetries}

The general Lie algebra of divergence-free vector fields mentioned in remark \ref{rem:Lie} has an important sub-algebra induced by
the group structure inherent to the finite symmetry transformations for system \eqref{eq:ODE}, definition \ref{def:finite_symm}. This is the Lie algebra of infinitesimal symmetries for system \eqref{eq:ODE}.

\begin{theorem}[Construction of new symmetries starting from known symmetries]
\label{thm:construct_symm}
Let $\etab_1$, $\etab_2$ be two smooth divergence-free infinitesimal contemporaneous symmetries for system \eqref{eq:ODE}. Namely, $\etab_1$ and $\etab_2$ satisfy the PDE system \eqref{eq:symm_Lie}. Then the smooth divergence-free  vector field defined by
$$\etab_3 := \Lie{\etab_1} \etab_2 = (\etab_1 \cdot \nabla) \etab_2 - (\etab_2 \cdot \nabla) \etab_1$$
is an infinitesimal contemporaneous symmetry for system \eqref{eq:ODE}.
\end{theorem}

\begin{proof}
The proof is a direct exercise and will be omitted.
\end{proof}

Theorem \ref{thm:construct_symm} is very important in practice as it allows for the construction of a succession of infinitesimal symmetries for system \eqref{eq:ODE} starting from two or more infinitesimal symmetries, providing thus a Lie algebra that  can be finite- or infinite-dimensional.  The only case when this construction is not possible is when the two starting infinitesimal symmetries commute: $\Lie{\etab_1} \etab_2 = \mbf{0}$. For this case we have an even better result:

\begin{theorem}[Construction of constants of motion starting from known commuting symmetries]
\label{thm:construct_CoM}
Let $\etab_1$, $\etab_2$ be two smooth divergence-free infinitesimal contemporaneous symmetries for system \eqref{eq:ODE} such that they are not linearly dependent (i.e., such that they are not proportional via a simple numerical factor). Assume these symmetries commute:
$\Lie{\etab_1} \etab_2 = \mbf{0}$. Then there are two cases:
\begin{enumerate}
\item If the vector product $\etab_1 \times \etab_2$ is not zero everywhere on $\Omega$ for all times, then $\etab_1 \times \etab_2 = \nabla K$ where $K$ is a non-trivial constant of motion for system \eqref{eq:ODE}, defined at least locally in $\Omega \times [0,T)$.
\item If the vector product $\etab_1 \times \etab_2$
is zero everywhere on $\Omega$ for all times, then $\etab_2 = K  \etab_1$, where $K$ is a constant of motion for system  \eqref{eq:ODE} whose gradient $\nabla K$ is not zero everywhere and is orthogonal to $\etab_1$ and $\etab_2$.
\end{enumerate}
\end{theorem}

\begin{proof}
(i) If $\etab_1 \times \etab_2$ is not zero everywhere on $\Omega$ for all times, we first want to find a scalar function $K: \Omega \times [0,T) \to \mathbb{R}$ such that $\etab_1 \times \etab_2 = \nabla K$. The compatibility condition for finding $K$ is \mbox{$\nabla \times (\etab_1 \times \etab_2) = \mbf{0}$} on $\Omega \times [0,T)$.  The function $K$ may not be globally defined, but at least it can be found to be smooth on some neighbourhood $B_\epsilon(\mbf{x}) \times [0,T)$, where $B_\epsilon(\mbf{x}) \subset \Omega$ is the ball of radius $\epsilon$ and centered at $\mbf{x}\in \Omega$. From elementary vector calculus identities we have $\nabla \times (\etab_1 \times \etab_2) = \etab_1 (\nabla \cdot \etab_2) - \etab_2 (\nabla \cdot \etab_1) + (\etab_2 \cdot \nabla) \etab_1 - (\etab_1 \cdot \nabla) \etab_2$. But $\nabla \cdot \etab_1=\nabla \cdot \etab_2=0$ since the fields are divergence-free. Therefore $\nabla \times (\etab_1 \times \etab_2) =- \Lie{\etab_1}\etab_2=\mbf{0}$. So $K$ exists locally. 

Second, we want to show that $\nabla C:=\nabla K$ satisfies \eqref{eq:gradCoM}. The LHS of this equation is:
$$\mbf{F}:=\frac{\partial \nabla K}{\partial t} + (\mbf{u}\cdot \nabla) \nabla K+ (\nabla \mbf{u})^T \nabla K = \frac{\partial}{\partial t}(\etab_1 \times \etab_2) + (\mbf{u}\cdot \nabla) (\etab_1 \times \etab_2)+ (\nabla \mbf{u})^T (\etab_1 \times \etab_2).$$
Looking at the first two terms of the last member, both the time derivative and the gradient satisfy the Leibniz product rule, so
$$\mbf{F} = \left(\frac{\partial \etab_1}{\partial t}+ (\mbf{u}\cdot \nabla) \etab_1\right) \times \etab_2 + \etab_1\times \left(\frac{\partial \etab_2}{\partial t}+ (\mbf{u}\cdot \nabla) \etab_2\right) + (\nabla \mbf{u})^T (\etab_1 \times \etab_2)\,.$$
Now, as $\etab_1$ and $\etab_2$ are infinitesimal symmetries they satisfy \eqref{eq:symm}, so we get
$$\mbf{F} = \left((\etab_1\cdot \nabla) \mbf{u}\right) \times \etab_2 + \etab_1\times \left((\etab_2\cdot \nabla) \mbf{u}\right) + (\nabla \mbf{u})^T (\etab_1 \times \etab_2)\,.$$
By inspection, the vector $\mbf{F}$ is an anticommutative algebraic bilinear function of the vectors $\etab_1, \etab_2$, such that $\mbf{F}\cdot \etab_1=\mbf{F}\cdot \etab_2=0$. Therefore, $\mbf{F} = M \etab_1 \times \etab_2$, where $M$ is some scalar function. By taking $\etab_1=(1,0,0)$ and $\etab_2=(0,1,0)$ we get readily $M=\nabla \cdot \mbf{u}$. As $\mbf{u}$ is divergence-free, we get $M=0$ and thus $\mbf{F}= \mbf{0}$. Therefore $\nabla K$ satisfies \eqref{eq:gradCoM}. 

Third and last, lemma \ref{lem:grad} establishes that $K$ is a constant of motion up to an additive function of time that can be found.

(ii) If the vector product $\etab_1 \times \etab_2$
is zero everywhere on $\Omega$ for all times, then $\etab_2 = K  \etab_1$ on $\Omega$ for all times, where by hypothesis $K$ is a scalar function that is not a simple numerical constant, thus $\nabla K$ is not zero everywhere. Let us apply the operator $\frac{\partial}{\partial t} + \Lie{\mbf{u}}$ to the equation $\etab_2 = K  \etab_1$. Using the Leibniz product rule for the Lie derivative we get
$$\frac{\partial \etab_2 }{\partial t} + \Lie{\mbf{u}} \etab_2 = K \left(\frac{\partial \etab_1}{\partial t} + \Lie{\mbf{u}} \etab_1\right) + \left(\frac{\partial K}{\partial t} + \Lie{\mbf{u}} K\right) \etab_1\,.$$
Recalling now that $\etab_1$ and $\etab_2$ are infinitesimal symmetries for system \eqref{eq:ODE}, namely that they satisfy \eqref{eq:symm}, we readily obtain
$\mbf{0} = \left(\frac{\partial K}{\partial t} + \Lie{\mbf{u}} K\right) \etab_1$ which, at the positions where $\etab_1$ is nonzero (generically this is almost everywhere), implies \mbox{$\frac{\partial K}{\partial t} + \Lie{\mbf{u}} K = 0$} and thus $K$ is a constant of motion for system \eqref{eq:ODE}. Now, $\nabla\cdot \etab_2 = \nabla\cdot \etab_1=0$ imply $\etab_1\cdot \nabla K =\etab_2\cdot \nabla K=0$.
\end{proof}

\begin{remark}
A dual version of theorem \ref{thm:construct_CoM}, stating in our notation that the cross product between the gradients of two constants of motion must be a symmetry, has been applied to fluid models in \cite{gibbon2010dynamics, gibbon2012stretching, kurgansky1987potential, kurgansky2000modified}. Also, a result analogous to theorem \ref{thm:construct_symm}, stating that the Lie derivative of a constant of motion along a symmetry is a constant of motion, is well known \cite{hojman1996non} and has been applied to fluid equations recently \cite{cheviakov2014generalized}.
\end{remark}

Theorem \ref{thm:construct_CoM} requires solely the existence of two commuting infinitesimal symmetries in order to construct a constant of motion, without recourse to Noether's theorem, namely without the need that system \eqref{eq:ODE} be derived from a Lagrangian or action principle that is invariant under some symmetry. See \cite{hojman1996non} for a  description of these so-called ``non-Noetherian'' symmetries. 

\section{The 3D Euler fluid equations; symmetries of the velocity field}
\label{sec:3}

Consider an inviscid, incompressible fluid, of unit mass density, moving in three spatial dimensions with smooth velocity field $\mbf{u}: \Omega \times [0,T) \to \mathbb{R}^3$. The notations $\mbf{x} := (x,y,z) = (x^1, x^2, x^3) \in \Omega$ and  $\mbf{u}(\mbf{x},t) := (u_{\mathrm{x}}(\mbf{x},t), u_{\mathrm{y}}(\mbf{x},t), u_{\mathrm{z}}(\mbf{x},t)) = (u^1(\mbf{x},t), u^2(\mbf{x},t), u^3(\mbf{x},t))$ will be used intechangeably. We will set either $\Omega=\mathbb{T}^3$, or $\Omega=\mathbb{R}^2 \times \mathbb{T}$,  or $\Omega=\mathbb{T}^2 \times \mathbb{R}$, where $\mathbb{T} := [0,2\pi)$ is the torus. 

The velocity field is thus divergence-free ($\nabla \cdot \mbf{u} = 0$) and satisfies the 3D Euler fluid equations:
\begin{equation}
 \label{eq:Euler}
\frac{\partial  \mbf{u}}{\partial t} + (\mbf{u} \cdot \nabla) \mbf{u} = - \nabla p\,, \qquad \text{for \,\, all}\quad \mbf{x} \in \Omega, \, t \in [0,T)\,,
\end{equation}
where the pressure $p$ is determined from the condition $\nabla \cdot \mbf{u} = 0$.  We will be interested in the solution of this nonlinear partial differential equation for the field $\mbf{u}(\mbf{x},t)$, under known smooth initial data $\mbf{u}(\mbf{x},0)=\mbf{u}_0(\mbf{x})$. At the same time we will be interested in the formalism of the previous sections, in terms of the solutions of the ODE system \eqref{eq:ODE}, namely the pathlines of the velocity field $\mbf{u}$. This velocity field may admit constants of motion and infinitesimal symmetries, namely solutions of equations \eqref{eq:CoM_Lie}--\eqref{eq:symm_Lie}.

The vorticity field $\bomega := \nabla \times \mbf{u}$ is clearly divergence-free and taking the curl of equation \eqref{eq:Euler} one obtains the well-known Helmholtz vorticity formulation \cite{helmholtz1858uber, oseledets1989new, arnold2008topological}: 
\begin{equation}
 \label{eq:Euler_vort}
\frac{\partial\bomega}{\partial t} + (\mbf{u} \cdot \nabla)\bomega - (\bomega \cdot \nabla) \mbf{u}  = \mbf{0}\,, \qquad \text{for \,\, all}\quad \mbf{x} \in \Omega, \, t \in [0,T)\,.
\end{equation}
From theorem \ref{thm:symm} it is then direct to prove:
\begin{lemma}[The vorticity is an infinitesimal symmetry of the flow]
Let $\mbf{u}: \Omega \times [0,T) \to \mathbb{R}^3$ be a smooth divergence-free velocity field satisfying the 3D Euler fluid equations  \eqref{eq:Euler}. Then the smooth divergence-free vorticity field $\bomega=\nabla \times \mbf{u}$ is an infinitesimal contemporaneous symmetry for the system \eqref{eq:ODE}. Namely, the vorticity vector field satisfies \eref{eq:symm} or, equivalently, \eref{eq:symm_Lie}.
\end{lemma} 

\begin{remark}
Historically, equation \eqref{eq:Euler_vort} was first obtained in its $2$-form formulation in \cite{helmholtz1858uber}. Similar transport equations for $3$-forms were found: helicity in 3D Euler equations \cite{oseledets1989new}, potential vorticity in atmospheric models \cite{ertel1942neuer, hoskins1985use, truesdell1960classical}. Of course, global helicity was shown to be conserved in 3D Euler equations much earlier \cite{moreau1961constantes, moffatt1969degree}. Further, the relation between Lie-advected $p$-forms and Cauchy invariants along with a historical account of how these have been rediscovered over and over again is documented in \cite{besse2017geometric}.
\end{remark}

If the flow admits another symmetry different from the vorticity, then it is possible to generate new symmetries and constants of motion, via theorems \ref{thm:construct_CoM} and \ref{thm:construct_symm}. In the following subsections we will consider in detail two examples where this construction is possible.

\subsection{Symmetries of steady 3D Euler flows}
By definition, steady 3D Euler flows are time-independent solutions $\mbf{u}(\mbf{x})$ 
 to the 3D Euler fluid equations \eqref{eq:Euler}. Of course, the initial data 
 must be carefully selected so that the solution remains static in time. The ODE system \eqref{eq:ODE} becomes autonomous. The vorticity $\bomega$ is thus time-independent, and from equation \eqref{eq:Euler_vort} and the definition of Lie derivative we get
$\Lie{\mbf{u}} \bomega = \mbf{0}$, for all $\mbf{x}\in\Omega$. 
On the other hand, since $\mbf{u}$ itself is time-independent  and $\Lie{\mbf{u}}\mbf{u}=\mbf{0}$, we deduce that $\mbf{u}$ is an infinitesimal symmetry for the ODE system \eqref{eq:ODE}. In conclusion, steady Euler flows possess two commuting symmetries, given by the static vector fields $\etab_1 = \mbf{u}$ and $\etab_2 = \bomega.$ Theorem \ref{thm:construct_CoM} gives:

\begin{corollary}
Let $\mbf{u}$ be a steady 3D Euler flow, with vorticity field $\bomega = \nabla \times \mbf{u}$. Then there exists a constant of motion $\alpha$ satisfying $\mbf{u} \times \bomega = \nabla \alpha$. 
\end{corollary}
\begin{proof}
This is a direct application of theorem \ref{thm:construct_CoM} on the symmetries $\etab_1 = \mbf{u}$ and $\etab_2 = \bomega$.
\end{proof}

\begin{remark}
The statement in this corollary is well known. In fact, $\alpha$ is the Bernoulli function, equal to $\alpha = p +\frac{1}{2} \mbf{u}\cdot \mbf{u}$, where $p$ is the pressure field. Moreover,  when $\nabla \times \mbf{u} = k \mbf{u}$ for some $k \in \mathbb{R}$ this corollary does not apply, in the sense that $\alpha$ becomes a numerical constant. In this case the velocity field is known as a Beltrami field \cite{arnold1966geometrie, childress1967construction}. As a result, Beltrami fields can be nonintegrable \cite{arnold2008topological}.
\end{remark}

As an example, consider a velocity field defined in Cartesian coordinates by $\mbf{u}(x,y,z) = (\cos z, a \sin z,0)$, where $a \in \mathbb{R}$ is a constant, and where either $\Omega=\mathbb{T}^3$ or $\Omega=\mathbb{R}^2\times\mathbb{T}$. The vorticity is $\bomega = -(a \cos z, \sin z, 0)$ and we find $\mbf{u}\times \bomega = \nabla \alpha\,,$ where
$\alpha = \frac{1}{4}(1-a^2) \cos 2z$
is a constant of motion (nontrivial when $|a| \neq 1$). Of course, when $|a|=1$ the flow is Beltrami, and a particular case of an ABC flow \cite{arnold1966geometrie, childress1967construction}. The pathlines of the velocity field satisfy the equations
\begin{equation}
\label{eq:char}
\dot{X} = \cos Z, \qquad \dot{Y} = a \sin Z, \qquad \dot{Z} = 0\,,
\end{equation}
and are simply straight lines (i.e., geodesics) on the sub-manifold $\Omega_Z$ defined by the equation $Z=\mathrm{const}$. For this simple flow, a general type of infinitesimal symmetry is thus given by
\begin{equation}
\label{eq:gen}
 \etab_{\mathrm{gen}} = f(\alpha) \mbf{u} + g(\alpha) \bomega\,,
\end{equation}
where $f, g$ are arbitrary scalar functions. Notice that this general symmetry has no $z$-component. Is it possible to find a new infinitesimal symmetry with a non-zero $z$-component? The answer depends on the topology of the spatial domain $\Omega$:\\

\noindent (i) If $\Omega = \mathbb{T}^3$ then $\Omega_Z = \mathbb{T}^2$. Notice that the values of $Z$ for which the pair $(\cos Z, a \sin Z)$ is incommensurate, cover densely the region $Z \in [0,2\pi).$ In these  cases, the pathlines are infinitely long and cover densely the $2$-torus $\Omega_Z$ parametrised by $(x,y).$ Let us call the sub-manifolds $\Omega_Z$ `ergodic' in such cases. On the other hand, embedded in the region $Z \in [0,2\pi)$ are an infinite number of points for which the pair $(\cos Z, a \sin Z)$ is commensurate, with closed pathlines on the $2$-torus $\Omega_Z$. We call the sub-manifolds $\Omega_Z$ regular in these cases. Let $\Omega_{Z_0}$ be an arbitrary regular sub-manifold, where $Z_0 \in (0,2\pi)$. Then, the open region ${\mathcal B}_\epsilon = \bigcup_{Z \in (Z_0-\epsilon,Z_0+\epsilon)} \Omega_Z$ contains, for arbitrary small enough $\epsilon>0$, ergodic sub-manifolds $\Omega_Z$. This is due to the fact that irrational numbers are uniformly distributed on the real line (equidistribution theorem).

We claim that it is impossible to find an infinitesimal symmetry continuous on ${\mathcal B}_\epsilon$, and defined globally there, with a non-zero component along the $z$-direction. To see this, suppose that such a symmetry exists, call it $\etab_0(x,y,z,t)$, and assume without loss of generality that its $z$-component is equal to $1$ at $t=0,x=0,y=0,z=Z_0:$ $\eta_0^z(0,0,Z_0,0) = 1\,.$ Let $\delta$ be such that $0<\delta<\epsilon$ and such that $\Omega_{Z_0+\delta} \in {\mathcal B}_\epsilon$ be ergodic. Recall that $\delta$ is arbitrarily small. Thus, the symmetry vector field $\delta  \etab_0(x,y,z,t)$ connects the point $P:=(0,0,Z_0) \in \Omega_{Z_0}$ with the point \mbox{$\widetilde{P}:=(\delta \eta_0^x(0,0,Z_0,0), \delta \eta_0^y(0,0,Z_0,0),Z_0+\delta) \in \Omega_{Z_0+\delta}\,.$} Using equation \eref{eq:char}, let us denote the pathline passing through $P$ by $(X(r),Y(r),Z_0) \in \Omega_{Z_0} \quad \forall r \in [0,1],$ with $(X(0),Y(0)) = (X(1),Y(1)) = (0,0)$ since $\Omega_{Z_0}$ is regular. Now, let us denote the pathline passing through $\widetilde{P}$ by $(\widetilde{X}(s),\widetilde{Y}(s),Z_0+\delta) \in \Omega_{Z_0+\delta}, \quad s \in (-\infty,\infty)$, which is not closed because $\Omega_{Z_0+\delta}$ is ergodic. But by definition the symmetry $\etab_0$ maps the closed pathline (of finite length) to the non-closed one (of infinite length). This is a contradiction, as it is not possible to continuously map these two pathlines. Therefore 
we conclude that the symmetry $\etab_0(x,y,z,t)$ is not continuous on ${\mathcal B}_\epsilon$.\\

\noindent (ii) If $\Omega = \mathbb{R}^2\times \mathbb{T}$ then $\Omega_Z = \mathbb{R}^2$ and the pathlines are just straight lines which extend without bound. All sub-manifolds are regular, and therefore the limitation of the previous case disappears. The following is an infinitesimal symmetry:
$\etab_{1}(x,y,z,t) = (-y , x\,a^2, a)$. The commutator of the vorticity with this symmetry is $\etab_{2}:=\Lie{\bomega}{\etab_{1}}=(1-a^2)(\sin z, a \cos z, 0)$. As  this is zero when $|a|=1$, we get two sub-cases:\\
\noindent (ii.1)  When $|a|=1$, i.e.~when the flow is Beltrami, we can write
\mbox{$\bomega \times \etab_{1} = \nabla C,$} giving the constant of motion $C=-a x \sin z + y \cos z$, which provides an equation for the pathlines on $\Omega_z$.\\
\noindent (ii.2) When $|a|\neq 1$, the symmetry $\etab_2$ is of the form \eqref{eq:gen} so there is no new information. In this case the following time-dependent symmetry does commute with the vorticity: $\etab_{3}(x,y,z,t) = (-y \,a^2 - t\,(1-a^2) \,a \sin z, x - t\,(1-a^2) \cos z, a)$. Thus, we can write
\mbox{$\bomega \times \etab_{3} = \nabla C',$} giving the constant of motion $\frac{1}{a} C'= x \sin z - a y \cos z - (1-a^2) t \sin z \cos z$ which provides an explicit time parameterisation for the pathlines on $\Omega_z$. 

A complete Lie algebra of symmetries can be produced, which we  omit for the sake of brevity.

\subsection{Symmetries of non-steady 3D Euler flows of stagnation-point type}
\label{subsec:3DGibbon}
We consider flows defined over the spatial domain $\Omega = \mathbb{T}^2 \times \mathbb{R}$ and use $(x,y)$ to parametrise $\mathbb{T}^2:=[0,2\pi)^2$ and $z$ to parametrise $\mathbb{R}.$ An exact solution of the 3D Euler equation \eqref{eq:Euler} is \cite{Gibbon1999497}:
\begin{equation}
\label{eq:u_Gibbon}
\mbf{u}(x,y,z,t) = (u_\mathrm{x}(x,y,t), u_\mathrm{y}(x,y,t), z\,\gamma(x,y,t))
\end{equation}
where, with reference to the plane $z=0$,  the ``horizontal'' velocity $\mbf{u}_\mathrm{h} := (u_\mathrm{x},u_\mathrm{y})$ is determined uniquely via the Helmholtz decomposition from the equations 
\begin{equation}
\label{eq:Helm}
\nabla_\mathrm{h} \cdot \mbf{u}_\mathrm{h} := \frac{\partial u_\mathrm{x}}{\partial x} + \frac{\partial u_\mathrm{y}}{\partial y} = -\gamma(x,y,t)\,, \qquad\qquad \nabla_\mathrm{h} \times \mbf{u}_\mathrm{h}:= \frac{\partial u_\mathrm{y}}{\partial x} - \frac{\partial u_\mathrm{x}}{\partial y} = \omega(x,y,t),
\end{equation}
where the scalar functions $\omega$ and $\gamma$ are, respectively, the out-of-plane vorticity and its vorticity stretching rate. These two scalars are the fundamental fields of this solution and satisfy the system
\begin{eqnarray}
\label{eq:vort}
\frac{\partial}{\partial t} \omega + \mbf{u}_\mathrm{h} \cdot \nabla_\mathrm{h} \omega  &=&  \gamma \, \omega\,,\\
\label{eq:gamma}
\frac{\partial}{\partial t} \gamma + \mbf{u}_\mathrm{h} \cdot \nabla_\mathrm{h} \gamma  &=&  2 \langle  \gamma^2 \rangle - \gamma^2\,,
\end{eqnarray}
where 
$\langle  f(\cdot,t) \rangle$ denotes the average of $f(x,y,t)$ over $(x,y)\in\mathbb{T}^2$ and is thus a function of time only. Equations \eqref{eq:Helm} imply $\langle  \gamma \rangle = \langle  \omega \rangle= 0$ at all times, which is consistent with \eqref{eq:vort}--\eqref{eq:gamma}. 

Due to the linear dependence of the $z$-component of the velocity on the $z$ coordinate, solutions of the form \eqref{eq:u_Gibbon} were termed ``of stagnation-point type'' by Ohkitani and Gibbon \cite{Ohkitani20003181}. As explained there, these solutions include well known and important subclasses such as the Burgers vortex \cite{burgers1948mathematical} and solutions found by Stuart \cite{stuart1987nonliear} and Childress \emph{et al.} \cite{childress1989blow}. Physically, stagnation points are interesting as they play a key role in practical situations where the flow is impinging perpendicularly on a wall, so particle pathlines must change their tangent vectors from normal to parallel to the wall and the velocity field is zero at one point, called ``stagnation point''. In our context, the wall is replaced with the symmetry plane $z=0$. On that plane, our velocity field is clearly parallel to the plane as its $z$-component goes to zero there, so in the generic case of non-uniform smooth velocity fields on $\mathbb{T}^2$, stagnation points may exist on the plane.

Because \eqref{eq:u_Gibbon} is an exact solution of \eqref{eq:Euler}, the full 3D vorticity 
\begin{equation}
\label{eq:vort_Gibbon}
\bomega = \nabla \times \mbf{u} = \left(z \frac{\partial \gamma}{\partial y}, -z \frac{\partial \gamma}{\partial x}, \omega \right)\,,
\end{equation}
is an infinitesimal symmetry for the non-autonomous system \eqref{eq:ODE}. Let us try to find another symmetry for this system, of the form $\etab = (0,0,\eta^{\mathrm{z}}(x,y,t))\,,$ where $\eta^{\mathrm{z}}$ is unknown. Plugging this into \eref{eq:symm_Lie} and using \eqref{eq:u_Gibbon} we get
\begin{equation}
\label{eq:eta_z}
 \partial_t \eta^{\mathrm{z}} + \mbf{u}_\mathrm{h} \cdot \nabla_\mathrm{h}\eta^{\mathrm{z}} - \gamma \, \eta^{\mathrm{z}} = 0\,.
\end{equation}
By inspection and comparing with \eqref{eq:vort} we see that a solution is $\eta^{\mathrm{z}} = \omega$, giving the symmetry
\begin{eqnarray}
 \label{eq:symm_omega}
\etab_1(x,y,t) := (0,0,\omega(x,y,t))\,.
\end{eqnarray}
This is the out-of-plane vorticity. Thus, by linearity, the in-plane vorticity is a symmetry:
\begin{eqnarray}
 \label{eq:symm_hor_lambda=0}
\etab_2(x,y,z,t):= \bomega(x,y,z,t) - \etab_1(x,y,t) =  \left(z \frac{\partial \gamma}{\partial y}, -z \frac{\partial \gamma}{\partial x}, 0 \right)\,.
\end{eqnarray}
From theorem \ref{thm:construct_symm}, successive Lie derivatives of $\etab_1$ along $\etab_2$ produces new symmetries:
\begin{equation}
 \label{eq:Lie_omega_Lie_omega_bomega}
\etab_3:=\Lie{\etab_1} \etab_2 = \left(\omega\,\frac{\partial \gamma}{\partial y}, - \omega\,\frac{\partial \gamma}{\partial x},
{z}\,\nabla_\mathrm{h}\gamma \times \nabla_\mathrm{h} \omega \right), \qquad
\etab_4:={\Lie{\etab_1}}\etab_3 =  \left(0,0, 2\,\omega\,\nabla_\mathrm{h}\gamma \times \nabla_\mathrm{h} \omega\right),
\end{equation}
with $\Lie{\etab_1} \etab_4= \mbf{0}\,,$ so ${\etab_1}$ and $\etab_4$ commute and are collinear (see equations  \eref{eq:Lie_omega_Lie_omega_bomega} and \eref{eq:symm_omega}), which by virtue of theorem \ref{thm:construct_CoM} gives the constant of motion $ C_1(x,y,t) :=  \nabla_\mathrm{h}\gamma \times \nabla_\mathrm{h} \omega\,.$

These results relate the fields $\gamma$ and $\omega$. In order to solve the evolution equations, it would be desirable to find an infinitesimal symmetry depending on the field $\gamma$ only. Let us try to find again a symmetry $\etab = (0,0,\eta^{\mathrm{z}}(x,y,t))\,,$ but this time look for solutions of equation \eref{eq:eta_z} of the form
\begin{equation}
\label{eq:ansatz}
\eta^{\mathrm{z}}(x,y,t) = \left(A(t) \, \gamma(x,y,t) + B(t)\right)^k\,,
\end{equation}
 where $k$ is a constant to be determined and $A(t), B(t)$ are functions of time to be determined. Replacing we get:
$k\, \dot A \,\gamma + k\,\dot B  + \left(\partial_t \gamma + \mbf{u}_\mathrm{h} \cdot \nabla_\mathrm{h} \gamma\right)\,k\,A  - \left(A \, \gamma + B\right) \, \gamma  = 0\,,
$ and using equation \eref{eq:gamma} we obtain
$k\, \dot A \,\gamma + k\,\dot B  + (2 \langle  \gamma^2 \rangle - \gamma^2)\,k\,A  - \left(A \, \gamma + B\right) \, \gamma  = 0$, which we can solve by eliminating the coefficients of $1$, $\gamma$ and $\gamma^2$ to get the system
\begin{equation*}
 - (k + 1)\,A = 0\,, \qquad k\,\dot A - B = 0\,,\qquad  k\,\dot B  + 2\,k \langle  \gamma^2 \rangle\,A = 0\,,
\end{equation*}
which gives the solution $k=-1, \quad B = - \dot A,$ and
$\ddot A - 2 \langle  \gamma^2 \rangle\,A = 0$,
so the new symmetry is
\begin{equation}
\label{eq:etab_5}
 \etab_5(x,y,t) = \left(0, 0, \frac{1}{A(t) \, \gamma(x,y,t) - \dot A(t)}\right)\,, \qquad \text{where} \qquad \ddot A - 2 \langle  \gamma^2 \rangle\,A = 0\,.
\end{equation}
As $\etab_5$ and $\etab_1$ are collinear, theorem \ref{thm:construct_CoM} gives the constant of motion $C_2(x,y,t):=\omega(x,y,t)(A(t) \, \gamma(x,y,t) - \dot A(t))$.
We want to find a constant of motion involving the field $\gamma$ only. So far the only two symmetries that depend just on $\gamma$  are $\etab_2$ and $\etab_5$.
Their Lie bracket is
\begin{equation}
\label{eq:etab_6}
\etab_6:=\Lie{\etab_5}\etab_2 = \frac{1}{A(t) \, \gamma(x,y,t) - \dot A(t)}\left(\frac{\partial \gamma}{\partial y}, - \frac{\partial \gamma}{\partial x}, 0 \right),
\end{equation}
proportional to $\etab_2$. Thus, theorem \ref{thm:construct_CoM} gives the constant of motion $C_3(x,y,z,t):= z (A(t) \, \gamma(x,y,t) - \dot A(t))$.  This constant determines the out-of-plane pathlines. Finally, notice that $\etab_6$ as well as $\etab_5$ do not depend on $z$. It is direct to check that they commute. As these symmetries are not proportional, theorem \ref{thm:construct_CoM}(i) allows for the construction of a new constant of motion such that
$\etab_6 \times \etab_5 = \nabla C_4$. We get, preliminarily,
$\nabla C_4 =  -{(A(t) \, \gamma(x,y,t) - \dot A(t))^{-2}} \left( \frac{\partial \gamma}{\partial x}, \frac{\partial \gamma}{\partial y}, 0\right)$, which is integrated to give
$C_4(x,y,t) =f(t)+ (A(t)[A(t) \, \gamma(x,y,t) - \dot A(t)])^{-1}$,
where $f(t)$ is to be determined. Using \eqref{eq:CoM} and the second equation in \eqref{eq:etab_5} we get $f'(t) = -[A(t)]^{-2}$. Thus
$$C_4(x,y,t) = \frac{1}{A(t)[A(t) \, \gamma(x,y,t) - \dot A(t)]} - \int_0^t \frac{1}{[A(s)]^2}~\mathrm{d}s\,.$$
This constant of motion allows one to solve for $\gamma$ along the pathlines and establish conditions for blowup, in a result that is equivalent to the solution presented in \cite{constantin2000euler} (see also \cite{mulungye2016lambda_0}). The constant of motion $C_2$ is then used to solve for $\omega$ along the pathlines. We will not go further in the solution here because this is a particular case of a  general model that we study in detail in the next section.

\section{A new family of models of 3D Euler flows of stagnation-point type: continuous symmetries and proof of finite-time blowup}
\label{sec:4}

The exact solution of the 3D Euler fluid equation considered in the previous subsection is a limited model for a general solution that has a symmetry plane at $z=0$: in a more general solution such as the one obtained from a numerical simulation on a domain $\Omega=\mathbb{T}^3$ \cite{bustamante2012interplay, Graf08, BusKerr08, HouLi08}, due to periodicity the velocity is not a linear function of $z$ anymore, and so the ansatz \eqref{eq:u_Gibbon}--\eqref{eq:gamma} will fail at reproducing the actual behaviour of the fields at the symmetry plane. 
In particular, it fails to predict modulations of the curvature of the pressure field, resulting in that $\left.\frac{\partial^2 p}{\partial z^2}\right|_{z=0}$ is a function of time only, not of the in-plane coordinates $(x,y)$. For this reason, a minimal generalisation of this model was introduced by \cite{mulungye2015lambda_m3/2}, so that this pressure curvature would show some modulation on the plane, while maintaining the Lie structure of infinitesimal symmetries for the model.  One can motivate this model in two seemingly different but equivalent ways: first, as done in \cite{mulungye2015lambda_m3/2}, we renounce the search for an exact solution of the 3D Euler equations on the whole spatial domain, dropping the ansatz \eqref{eq:u_Gibbon}, and simply focus on the 3D Euler equations restricted to the symmetry plane $z=0$, where the required closure for the term $\left.\frac{\partial^2 p}{\partial z^2}\right|_{z=0}$ is constructed via a ``sparse'' modification, as in SINDy models \cite{rudy2017data}: $\left.\frac{\partial^2 p}{\partial z^2}\right|_{z=0} := - 2 \langle \gamma^2\rangle + \lambda (\gamma^2 - \langle \gamma^2\rangle)$, with a free \emph{dimensionless} parameter $\lambda \in \mathbb{R}$. The horizontal velocity field $\mbf{u}_\mathrm{h} := (u_\mathrm{x}, u_\mathrm{y})|_{z=0}$ is the relevant unknown. This leads to equations \eqref{eq:Helm}--\eqref{eq:vort} as before, while the scalar function $\gamma$ satisfies now
\begin{eqnarray}
\frac{\partial}{\partial t} \gamma + \mbf{u}_\mathrm{h} \cdot \nabla_\mathrm{h} \gamma  &=&  2 \langle  \gamma^2 \rangle - \gamma^2 + \lambda \left(\langle  \gamma^2 \rangle -\gamma^2 \right)
\label{eq:gamma_lambda}
\end{eqnarray}
where, as usual,
$\langle  f(\cdot,t) \rangle$ denotes the average of $f(x,y,t)$ over $(x,y)\in\mathbb{T}^2$, and $\langle \gamma\rangle=\langle \omega\rangle = 0$. Notice that the case $\lambda=0$ corresponds to the exact solution of the 3D Euler equations discussed in the previous subsection. In this first motivation, the main equations of the model are  \eqref{eq:Helm}, \eqref{eq:vort} and \eqref{eq:gamma_lambda}. They can be solved along pathlines on the symmetry plane $(x,y)\in\mathbb{T}^2$ using continuous symmetry methods. However, in this motivation the role of the $z$ coordinate is unclear. Thus we introduce and adopt a second, new motivation, completely equivalent to the first one from the viewpoint of equations \eqref{eq:Helm}, \eqref{eq:vort} and \eqref{eq:gamma_lambda}, but where we insist on an exact solution, this time of modified 3D Euler equations on the whole spatial domain $\Omega = \mathbb{T}^2 \times \mathbb{R}$:
\begin{equation}
 \label{eq:model}
\frac{\partial}{\partial t} \mbf{u} + \mbf{u} \cdot \nabla \mbf{u} = - \nabla p - \hat{\mbf{e}}_\mathrm{z}\,\lambda \,u_\mathrm{z} \frac{\partial u_\mathrm{z}}{\partial z}\,, \qquad \nabla \cdot \mbf{u} = 0\,,
\end{equation}
where $\hat{\mbf{e}}_\mathrm{z}$ is the unit vector in direction $z$. 
An exact solution of this new model is given by ansatz \eqref{eq:u_Gibbon}--\eqref{eq:vort} along with \eqref{eq:gamma_lambda}, and the case $\lambda=0$ corresponds exactly to the original 3D Euler equations. 
The new term in  \eqref{eq:model} is a continuous deformation of the convective derivative term, a familiar procedure in more tractable 1D fluid models \cite{morlet1998further, okamoto2000some, okamoto2008generalization, matsumoto2016one}. It keeps the variational structure intact (e.g., symmetries and constants of motion), which allows for the study of the role of initial conditions on the finite-time singularity of the solutions, even for the original problem $\lambda=0$.


\subsection{Symmetries and constants of motion for the new model}
The new model comes at a price: the full vorticity $\bomega = \nabla \times \mbf{u}$ is not an infinitesimal symmetry anymore because the extra forcing term is not a pure gradient. But all other symmetries found in the previous subsection persist in the new model, perhaps with a dependence on the parameter $\lambda$, and the method for finding these symmetries is quite similar. The relevant symmetries are:
\begin{eqnarray}
\label{eq:symm_vort}
\etab_1(x,y,t) = (0,0,\omega(x,y,t)), \qquad \etab_2(x,y,z,t) = \left(z^{2\,\lambda+1} \,\frac{\partial \gamma}{\partial y}, - z^{2\,\lambda+1} \,\frac{\partial \gamma}{\partial x}, 0  \right)\,,\\
 \label{eq:symm_Gibbon_model}
  \etab_5 (x,y,t) = \left(0,0,\left(A(t) \,\gamma(x,y,t) - \frac{\dot A(t)}{1+\lambda}\right)^{-\frac{1}{1+\lambda}}\right)\,, \qquad \lambda \neq -1\,,\\
\label{eq:ODE_A_model}
 \ddot{A}(t) - (\lambda+1)(\lambda+2) \langle  \gamma^2 \rangle \,A(t)= 0\,.
\end{eqnarray}
From here on we assume $\lambda \neq -1$. The case $\lambda=-1$ has an elementary solution and will be omitted.

As $\etab_1$ and $\etab_5$ are collinear, we obtain the constant of motion
$$C_2(x,y,t) := \omega(x,y,t) \left(A(t) \,\gamma(x,y,t) - \frac{\dot A(t)}{1+\lambda}\right)^{\frac{1}{1+\lambda}}\,,\qquad \lambda \neq -1\,.
$$
Also, the Lie derivative of $\etab_2$ along $\etab_5$ is another symmetry (from theorem \ref{thm:construct_symm}) and satisfies 
$\Lie{\etab_5} \etab_2 = (2\,\lambda + 1) \,z^{-1} \left(A(t) \,\gamma(x,y,t) - \frac{\dot A(t)}{1+\lambda}\right)^{-\frac{1}{1+\lambda}} \,\etab_2$. As this is proportional to $\etab_2$, from theorem \ref{thm:construct_CoM}(ii) we get the constant of motion (also valid for $\lambda = -1/2$)
\begin{equation}
\nonumber
C_3(x,y,z,t) := z\,\left(A(t) \,\gamma(x,y,t) - \frac{\dot A(t)}{1+\lambda}\right)^{\frac{1}{1+\lambda}} \,,\qquad \lambda \neq -1\,.
\end{equation}
Defining now the re-scaled symmetry
${\etab}_6(x,y,z,t) := \left[C_3(x,y,z,t)\right]^{-1-2\,\lambda}\,\etab_2(x,y,z,t)$, we obtain $\Lie{\etab_5} \etab_6 = \mbf{0}$. As $\etab_5$ and $\etab_6$ are not collinear, theorem \ref{thm:construct_CoM}(i) gives
$\etab_6 \times \etab_5 = \nabla C_4\,,$ where 
\begin{equation*}
C_4(x,y,t) = \frac{1}{A(t) \left(A(t) \,\gamma(x,y,t) - \frac{\dot A(t)}{1+\lambda}\right)} + f(t)\,,
\end{equation*}
and $f(t)$ is to be determined. To determine it we simply compute the convective time derivative of $C_4$ using \eqref{eq:CoM}, \eqref{eq:u_Gibbon} and \eref{eq:ODE_A_model}, giving $\dot f(t) + (\lambda + 1)/(A(t))^2 = 0$, which is integrated to yield the final form of the constant of motion $C_4$:
\begin{equation}
\label{eq:D_model}
 C_4(x,y,t) = \frac{1}{A(t) \left(A(t) \,\gamma(x,y,t) - \frac{\dot A(t)}{1+\lambda}\right)} -  (\lambda + 1) \int_0^t \frac{1}{[A(s)]^2}~\mathrm{d}s\,, \qquad \lambda \neq -1\,.
\end{equation}

\subsection{Solving for the fields along the pathlines}

The constant of motion \eref{eq:D_model} allows one to solve for $\gamma$ along pathlines $(X(t),Y(t),Z(t))$, solutions of the system \eref{eq:ODE}. Setting initial conditions at $t=0$ and noting that $A$ satisfies the second-order ODE \eref{eq:ODE_A_model}, we set for simplicity $A(0) = 1, \quad \dot A(0) = 0\,.$ Defining $\gamma_0(x,y) = \gamma(x,y,0)$ we get:
$C_4(X(t),Y(t),t) = C_4(X_0,Y_0,0) = 1/\gamma_0(X_0,Y_0)\,, $ and solving for $\gamma(X(t),Y(t),t)$ we get
\begin{equation}
 \label{eq:sol_gamma}
\gamma(X(t),Y(t),t) = \frac{\gamma_0(X_0,Y_0)}{A(t)^2\left(1 + (\lambda+1)\,\gamma_0(X_0,Y_0) \int_0^t\frac{1}{A(s)^2}\mathrm{d}s \right)} + \frac{\dot A(t)}{(\lambda+1)\,A(t)}\,,\quad \lambda \neq -1\,.
\end{equation}
This solution either blows up in a finite time or remains regular forever, depending on $\lambda$ and $\gamma_0$.

Evaluating $C_2(x,y,t)$ at the pathline at $t=0$ we get $C_2(x,y,t) = \omega_0(X_0,Y_0) \left[\gamma_0(X_0,Y_0)\right]^{\frac{1}{1+\lambda}}$ and using equation \eref{eq:D_model} to eliminate $\gamma$ we obtain
\begin{equation}
 \label{eq:sol_omega}
\omega(X(t),Y(t),t) = \omega_0(X_0,Y_0) \left[A(t)\,\left(1 + (\lambda+1)\,\gamma_0(X_0,Y_0) \int_0^t\frac{1}{A(s)^2}\mathrm{d}s \right)\right]^{\frac{1}{\lambda+1}}\,.
\end{equation}

A crucial quantity is the Jacobian of the back-to-labels transformation in the plane $z=0$ parameterised by $(x,y) \in \mathbb{T}^2.$  We define
\begin{equation*}
 J(t;X_0,Y_0) = \det\left(\frac{\partial(X(t),Y(t))}{\partial(X_0,Y_0)}\right)\,,
\end{equation*}
where $X(t),Y(t)$ are the $x$- and $y$-components of the pathline with initial positions $X_0, Y_0.$ A computation of the area of the 2-torus $\mathbb{T}^2$ gives an identity to be satisfied by the Jacobian:
$\int_{\mathbb{T}^2}\mathrm{d}x\mathrm{d}y = 4\pi^2 = \int_{\mathbb{T}^2} J(t;X_0,Y_0)\mathrm{d}X_0\mathrm{d}Y_0\,,$ or in terms of averages:
\begin{equation}
 \label{eq:aver_J}
\langle J(t;\cdot)\rangle_0 = 1\,,
\end{equation}
where $\langle  f \rangle_0$ denotes an average over the `label' variables $X_0,Y_0$. Noticing that $J$ is positive as long as there is no blowup, we can interpret $J(t;X_0,Y_0)$ as a probability density on the 2-torus.

From the fact that $\nabla_\mathrm{h} \cdot \mbf{u}_\mathrm{h} = -\gamma$ we readily obtain an evolution equation for $J$:
\begin{equation*}
\frac{ \dot J(t;X_0,Y_0)}{J(t;X_0,Y_0)} = -  \gamma(X(t),Y(t),t) \quad \Longrightarrow \quad J(t;X_0,Y_0) = \exp\left({-\int_0^t \gamma(X(s),Y(s),s)\mathrm{d}s}\right)\,,
\end{equation*}
and using equation \eref{eq:sol_gamma} we get 
$J(t;X_0,Y_0) = \left[A(t)\,\left(1 + (\lambda+1)\,\gamma_0(X_0,Y_0) \int_0^t\frac{1}{A(s)^2}ds \right)\right]^{-\frac{1}{\lambda+1}}$.
Finally, replacing this solution into condition \eref{eq:aver_J} and defining the increasing function
\begin{equation}
\label{eq:A_from_S}
S(t) =  \int_0^t\frac{1}{A(s)^2}\mathrm{d}s  \qquad \Longrightarrow \qquad A(t) = [\dot S(t) ]^{-1/2}\,,
\end{equation}
we get a remarkable first-order ODE for $S(t)$:
\begin{equation}
 \label{eq:formula_S}
 \dot S(t)  = \left\langle \bigg[1 + (\lambda+1)\,\gamma_0(\cdot) \,S(t)\bigg]^{-\frac{1}{\lambda+1}}\right\rangle_{0}^{-2(\lambda+1)}\,,\quad S(0)=0\,.
\end{equation}
Once the solution of this ODE is found, then $A(t)$ can be obtained and so the basic fields $\gamma, \omega$ can be found explicitly along the pathlines. This ODE contains all the needed information to determine whether a blowup occurs in a finite time. The initial data for $\omega$ plays no role.

\subsection{Singularity time and blowup asymptotics}
Combining Eqs.~\eqref{eq:sol_gamma}--\eref{eq:sol_omega} with Eq.~\eqref{eq:A_from_S} we obtain $\gamma$ and $\omega$ along the pathlines in terms of $S(t)$:
\begin{eqnarray}
\label{eq:sol_gamma_S}
 \gamma(X(t),Y(t),t) &=&  \frac{\dot S(t) \,\gamma_0(X_0,Y_0)}{1 + (\lambda+1)\,\gamma_0(X_0,Y_0) \,S(t)} - \frac{1}{2(\lambda+1)} \frac{\ddot S(t)}{\dot S(t)}\,,\\
\label{eq:sol_omega_S} \omega(X(t),Y(t),t) &=& \omega_0(X_0,Y_0) \left[\frac{1 + (\lambda+1)\,\gamma_0(X_0,Y_0) \,S(t)}{\dot S(t)^{1/2}}\right]^{\frac{1}{\lambda+1}}\,.
\end{eqnarray}
A singularity occurs at the first time $T_*$ when $1 + (\lambda+1)\,\gamma_0(X_0,Y_0) \,S(T_*)=0$ for some $(X_0,Y_0) \in \mathbb{T}^2$, because $\dot S \geq 0$. The behaviour of $\dot S$ and $\ddot S$ at $T_*$ may affect the blowup asymptotics. From \eqref{eq:sol_gamma_S} we get $d\gamma/d\gamma_0 > 0$, and thus the level sets of $\gamma$ keep the same ordering at all times. Notice that system \eqref{eq:formula_S}--\eqref{eq:sol_omega_S} combined with \eqref{eq:Helm} to solve for the in-plane velocity, allows for general (not necessarily continuous) initial distributions $\gamma_0, \omega_0$: bounded and of zero mean on $\mathbb{T}^2$.
\begin{theorem}[Singularity time]
Let $\gamma_0$ be nonzero, bounded on $\mathbb{T}^2$ and of zero mean. Then there is a first time $T_\mathrm{*} \in (0,\infty) \cup \{\infty\}$ at which a singularity occurs, defined by the condition
\begin{equation}
\label{eq:singularity_time_1}
S(T_*) = S_*\,,\qquad 
\text{where} \qquad 
 0< \frac{1}{S_*} := -(\lambda+1)\begin{cases}\displaystyle {\sup_{(X_0,Y_0) \in \mathbb{T}^2}{\gamma_0(X_0,Y_0)}},& \lambda  < - 1\,,\\
						   \displaystyle  \inf_{(X_0,Y_0) \in \mathbb{T}^2}{\gamma_0(X_0,Y_0)},& -1 < \lambda\,.
                             \end{cases}
\end{equation}
Moreover, an explicit formula for this singularity time (or blowup time) is
\begin{equation}
 \label{eq:T_sing}
 T_* = \int_0^{S_*} \left\langle \bigg[1 + (\lambda+1)\,\gamma_0(\cdot) \,S\bigg]^{-\frac{1}{\lambda+1}}\right\rangle_{0}^{2(\lambda+1)}\,\mathrm{d}S\,,
\end{equation}
\end{theorem}
\begin{proof}
Because $\gamma_0$ is nonzero, bounded and of zero mean, it has a positive global maximum and a negative global minimum on $\mathbb{T}^2.$ Noting that $S(0)=0$, it follows from \eqref{eq:formula_S}--\eqref{eq:singularity_time_1} that $\dot S(t) > 0$ when $S(t)<S_*$ and $\lambda  \neq -1\,.$ Thus, eventually $S(t)=S_*$ for some $t = T_*>0$, which we call the singularity time. Integrating the ODE \eref{eq:formula_S} using separation of variables we obtain \eqref{eq:T_sing}.
\end{proof}

\begin{remark}
We have in general $S_* < \infty$ when $\lambda \neq -1,$ but we do not know whether $T_*$ is finite or infinite: this will depend on the type of initial condition as well as on the value of $\lambda.$ The singularity time $T_\mathrm{*}$ is finite if and only if the integral \eqref{eq:T_sing} converges. Thus, in order to determine the singular behaviour of the system, we need to look at the behaviour of the solution of the ODE \eref{eq:formula_S} near $S=S_*$. 
\end{remark}

The following results will help us identify and classify some finite-time singularity cases:

\begin{lemma}
\label{lem:general_blowup}
{If $S(t)$ satisfies $\displaystyle \lim_{t \to T_*} \dot S(t) > 0$ then $T_* < \infty$, i.e.~the singularity occurs at a finite time.}
\end{lemma}

\begin{proof}
  The function $\dot S(t)$ is nothing but the multiplicative inverse of the integrand in \eref{eq:T_sing}. By definition, $\dot S(t)>0$ if $0\leq t < T_*$. Thus if $\displaystyle \lim_{t \to T_*} \dot S(t) > 0$ then the integral converges.
\end{proof}

\begin{lemma}[Formula for $\ddot S(t)$] The second time derivative of the function $S(t)$, solution of \eqref{eq:formula_S}, is
\begin{equation}
\label{eq:ddot_S}
\ddot S(t) = 2(\lambda+1)\left[\dot S(t)\right]^{2+\frac{1}{2(\lambda+1)}}\left\langle \gamma_0(\cdot)  \bigg[1 + (\lambda+1)\,\gamma_0(\cdot) \,S(t)\bigg]^{-\frac{\lambda+2}{\lambda+1}}\right\rangle_{0}
\end{equation}
\end{lemma}
\begin{proof}
The result follows directly after time differentiation of equation \eqref{eq:formula_S}.
\end{proof}

\begin{definition}[Positions of supremum and infimum of $\gamma_0$]
\label{def:sup_inf_pos}
Let the initial condition $\gamma_0$ be nonzero, bounded on $\mathbb{T}^2$ and of zero mean. We define this function's supremum $\gamma_+$ and infimum $\gamma_-$, and their respective positions $(X_+,Y_+)$ and $(X_-,Y_-)$, as the (not necessarily unique) points on $\mathbb{T}^2$ such that
$$\gamma_+:=\gamma_0(X_+,Y_+) = \sup_{(X_0,Y_0) \in \mathbb{T}^2}{\gamma_0(X_0,Y_0)}\,,  \qquad \qquad \gamma_- :=\gamma_0(X_-,Y_-) = \inf_{(X_0,Y_0) \in \mathbb{T}^2}{\gamma_0(X_0,Y_0)}\,.$$
\end{definition}

As the nonnegative terms $[1 + (\lambda+1)\,\gamma_0(\cdot) \,S(t)]$, appearing in the averages over $\mathbb{T}^2$ in equations \eqref{eq:formula_S}, \eqref{eq:T_sing} and \eqref{eq:ddot_S}, are close to zero near the singularity time and near the position of supremum/infimum of $\gamma_0$, the key idea to understand singularity time and blowup asymptotics is to look at whether the relevant exponents are negative or not. In what follows, we call ``unconditional results'' those results that do not depend on the specific profile of $\gamma_0$ near the position of its supremum or infimum. We call ``conditional results'' everything else.

\subsection{Unconditional blowup results: singularity time and asymptotics}

\begin{theorem}[Unconditional finite-time blowup for $\lambda<-1$]
\label{thm:blowup_lambda<-1}
{Let $\lambda < -1$. Let the initial condition $\gamma_0$ be nonzero, bounded on $\mathbb{T}^2$ and of zero mean. Then the singularity occurs at a finite time, i.e., $T_* < \infty.$}
\end{theorem}

\begin{proof}
 If $\lambda < -1$ and $\gamma_0$ is bounded then the integrand of the area integral in equation \eqref{eq:formula_S}, at $S=S_*$, is bounded and non-negative.  As $\gamma_0$ is of zero mean and bounded, this integrand is strictly positive in a subset of $\mathbb{T}^2$ of nonzero measure. Therefore the area integral is positive and bounded and thus $0 < \displaystyle \lim_{t \to T_*} \dot S(t)$. From lemma \ref{lem:general_blowup} this implies $T_* < \infty$. 
\end{proof}

\begin{theorem}[Unconditional blowup asymptotics for $\lambda < -1$] 
\label{thm:asymptotics_uncond}
Let $\lambda < -1$. Let the initial condition $\gamma_0$ be nonzero, bounded on $\mathbb{T}^2$ and of zero mean. Let $(X_+,Y_+)\in \mathbb{T}^2$ be any position of the supremum of $\gamma_0$. Then the stretching rate $\gamma$ blows up at time $T_*$ at the point $(X(T_*),Y(T_*)) \in \mathbb{T}^2,$ corresponding to the pathline with initial position $(X(0),Y(0)) = (X_+,Y_+)$. Explicitly we have:
$$\displaystyle \sup_{(x,y)\in \mathbb{T}^2} \gamma(x,y,t) \sim \gamma(X(t),Y(t)) \sim \frac{1}{|\lambda+1|} \left(T_*-t\right)^{-1}\,, \qquad \text{as} \qquad t \to T_*\,. $$
As for the out-of-plane vorticity $\omega,$ it blows up at time $T_*$ if and only if the initial vorticity satisfies $\omega_0(X_+,Y_+) \neq 0.$ The blowup occurs at the same point at which $\gamma$ blows up and we have, explicitly:
$$\displaystyle \sup_{(x,y)\in \mathbb{T}^2} |\omega(x,y,t)| \sim |\omega(X(t),Y(t))| \sim K_+\,\left(T_*-t\right)^{-\frac{1}{|\lambda+1|}}\,,\qquad \text{as} \qquad t \to T_*\,,$$
where $K_+ = \frac{{|\omega_0(X_+,Y_+)|}}{|\lambda+1|^{\frac{1}{|\lambda+1|}}} \left\langle \left[{\gamma_0(X_+,Y_+)} - {\gamma_0(\cdot)} \right]^{\frac{1}{|\lambda+1|}}\right\rangle_{0}^{-1}.$
At any other pathline not starting at the position of the supremum of $\gamma_0$, the stretching may blow up when $\lambda<-2$, but it would do so asymptotically more slowly than its supremum, while the vorticity at any other pathline is finite at all times.
\end{theorem}
\begin{proof}
We first note that, from the proof of theorem \ref{thm:blowup_lambda<-1} we have $T_*<\infty$ and $0< \dot S(T_*) < \infty$. Thus, looking at the formulae for $\gamma$ and $\omega$ along the pathlines, \eqref{eq:sol_gamma_S}--\eqref{eq:sol_omega_S}, we see that the blowup depends on the factor $1 + (\lambda+1)\,\gamma_0(X_+,Y_+) \,S(t) = 1-S(t)/S_*$, which has a simple zero as a function of $t$ at $t=T_*$, because $S(t) \sim S_*-\dot S(T_*) (T_*-t)$. Thus, this factor will produce a blowup only if the initial position of the pathline is at a supremum of  $\gamma_0$. Regarding the blowup of $\gamma$, the only question remaining is whether the term $\ddot S$ is finite, or at least asymptotically smaller than the first term in \eqref{eq:sol_gamma_S}. From equation \eqref{eq:ddot_S} we see that if $-2 \leq \lambda < -1$ then the exponent $-\frac{\lambda+2}{\lambda+1}$ is nonnegative and thus $\ddot S(T_*)$ is finite. For $\lambda < -2$, this exponent 
is negative but we can bound:
$$\bigg[1 + (\lambda+1)\,\gamma_0(\cdot) \,S(t)\bigg]^{-\frac{\lambda+2}{\lambda+1}} \leq \bigg[1 + (\lambda+1)\,\gamma_+ \,S(t)\bigg]^{-\frac{\lambda+2}{\lambda+1}} = S_*^{\frac{\lambda+2}{\lambda+1}}(S_*-S(t))^{-\frac{\lambda+2}{\lambda+1}},$$
leading to $|\ddot S(t)| \leq c (S_*-S(t))^{-\frac{\lambda+2}{\lambda+1}}$, a bound that blows up. But the first term in \eqref{eq:sol_gamma_S} blows up as $(S_*-S(t))^{-1}$, which is asymptotically faster than $|\ddot S(t)|$ because \mbox{$-\frac{\lambda+2}{\lambda+1} > -1$.}
Thus, the blowup of $\gamma$ in \eqref{eq:sol_gamma_S} is dictated by the first term, with the desired result. As for the blowup of $\omega$ in \eqref{eq:sol_omega_S} at the pathline starting at $(X_+,Y_+)$, notice that \mbox{$\dot S(T_*) = {\gamma_0(X_+,Y_+)^{-2}}\left\langle \left[{\gamma_0(X_+,Y_+)} - {\gamma_0(\cdot)} \right]^{-\frac{1}{\lambda+1}}\right\rangle_{0}^{-2(\lambda+1)}$.} The desired result follows after direct algebraic manipulations. At any other pathline, $\omega$ remains finite.
\end{proof}

\subsection{Conditional blowup results for singularity time and asymptotics, in terms of the type of initial conditions}

For $-1 < \lambda$ the singularity time and blowup asymptotics depend on the initial condition for the stretching rate, specifically on the local profile of $\gamma_0$ near the positions $(X_+,Y_+)$ and $(X_-,Y_-)$ of definition \ref{def:sup_inf_pos}. For simplicity we consider only one type of initial condition.

\begin{definition}[Generic type of function]
{Let the function $f: \mathbb{T}^2 \to \mathbb{R}$ be nonzero, bounded and of zero mean. This function is said to be of generic type if it attains its supremum at a single point $(X_+,Y_+)$ and its  infimum at a single point $(X_-,Y_-)$, such that the matrix of second derivatives $D^2 f$ is negative definite at $(X_+,Y_+)$ and positive definite at $(X_-,Y_-)$.} In words, $f$ is of generic type if it has a unique global maximum and a unique global minimum, and its level contours near these points are asymptotically ellipses.
\end{definition}

Theorem \ref{thm:blowup_lambda<-1} already establishes a finite-time blowup if $\lambda < -1.$ For other values of $\lambda$ we have:

\begin{theorem}[Conditions on singularity time for generic initial conditions and $-1 < \lambda$]
\label{thm:blowup_conditions}
Let the initial condition $\gamma_0$ be nonzero, bounded on $\mathbb{T}^2$, of zero mean and of generic type. These statements follow:
\begin{enumerate}

 \item If $-1< \lambda \leq -1/2$ then the flow is regular at all times: $T_*=\infty$, and $ \lim_{t \to T_*} \dot S(t) = 0.$
 \item If $-1/2 < \lambda \leq 0$ then a singularity occurs at a finite time: $T_* < \infty$, and $\dot S(T_*) = 0.$
 \item If $0 < \lambda$ then $0 <  \lim_{t \to T_*} \dot S(t)$ and thus a singularity occurs at a finite time: $T_* < \infty.$
\end{enumerate}
\end{theorem}
\begin{proof}
By looking at the RHS of \eref{eq:formula_S}, the exponent inside the area integral is negative, so the contribution to the singularity when $S\sim S_*$, if any, must come from the region in $\mathbb{T}^2$ that is near the global minimum $(X_-,Y_-)$ of definition \ref{def:sup_inf_pos}. After a rotation and rescaling we get:
\begin{equation}
\label{eq:dot_S_generic}
 \dot S(t) \approx \left[\frac{1}{2\pi \sqrt{\det D^2 \gamma_0(X_-,Y_-)}} \int_0^\epsilon (1-S(t)/S_* + (\lambda+1) S(t) r^2/2)^{-1/(\lambda+1)} r \mathrm{d}r\right]^{-2(\lambda+1)}\,,
\end{equation}
where $\epsilon>0$ is small but fixed and $D^2 \gamma_0(X_-,Y_-)$ is the matrix of second derivatives of $\gamma_0$ evaluated at the global minimum. We first consider the case $0<\lambda$. There, the integral over $r$ converges as $S\to S_*$ so we deduce $\dot S(T_*) > 0$ which by lemma \ref{lem:general_blowup} implies $T_*<\infty$, thus proving statement \emph{(iii)}.  Second, for $-1<\lambda \leq 0$ the integral is elementary, and keeping the leading term (at $r=0$) we get these asymptotic expressions as $S\to S_*$:
\begin{equation}
\label{eq:asymptotics_S}
 \dot S(t) \sim \begin{cases}
                     \left[4\pi^2 S_*^2 \lambda^2 \, \det D^2 \gamma_0(X_-,Y_-)\right]^{\lambda+1} \left(1-S(t)/S_*\right)^{-2\,\lambda}\,, & \quad -1 < \lambda < 0\,, \\
		    4\pi^2 S_*^2 \left[\det D^2 \gamma_0(X_-,Y_-)\right] \left(\log[1-S(t)/S_*]\right)^{-2}\,, & \quad \lambda = 0\,.
                   \end{cases}
\end{equation}
In all these cases $\dot S = 0$ at $S=S_*$. Recalling that $T_* = \int_0^{S_*} \dot S^{-1}\mathrm{d}S$, the above expressions can be used to determine if $T_*$ converges or not. Clearly, $T_* = \infty$ for $-1<\lambda \leq -1/2$, so statement \emph{(i)} is proven. In contrast, $T_* < \infty$ for $-1/2 < \lambda < 0$ and also for $\lambda=0$, thus statement \emph{(ii)} is proven.
\end{proof}

\begin{remark}
The statements of this theorem depend sensitively on the profile of $\gamma_0$ near the position of its infimum. For example, if one considers an initial condition such that $D^2\gamma_0(X_-,Y_-)$ has one positive eigenvalue and one zero eigenvalue, then by an elementary calculation the statements of theorem \ref{thm:blowup_conditions} get modified by simply changing the respective ranges of values of $\lambda$: in  (i) the new range is $-1 <\lambda \leq 0$; \mbox{in (ii) the new range} is $0<\lambda \leq 1$; in (iii) the new range is $1<\lambda$. Thus, the case of exact solution of 3D Euler equations, $\lambda=0$, suddenly becomes regular ($T_*=\infty$) for this type of initial condition.
\end{remark}

Theorem \ref{thm:asymptotics_uncond} shows unconditional blowup asymptotics for $\lambda < -1$. Theorem \ref{thm:blowup_conditions} shows there is no blow up with generic initial conditions for $-1 < \lambda \leq 1/2$. As for other values of $\lambda$ we have:

\begin{theorem}[Conditions on blowup asymptotics for generic initial conditions and $-1/2 < \lambda$]
\label{thm:conditions_-1/2<lambda}
Let $-1/2 < \lambda$. Let the initial condition $\gamma_0$ be nonzero, bounded on $\mathbb{T}^2$, of zero mean and of generic type.  Let $(X_-,Y_-)$ and $(X_+,Y_+)\in \mathbb{T}^2$ be the respective positions of the infimum $\gamma_-$ and supremum $\gamma_+$ of $\gamma_0$. Let $\Delta:=\det D^2 \gamma_0(X_-,Y_-)>0$ be the determinant of the matrix of second derivatives of $\gamma_0$ at the position of its infimum. Then the stretching rate $\gamma$ 
has the following asymptotic behaviour as $t \to T_*$:\\

\noindent (i) At the pathline with initial position $(X(0),Y(0)) = (X_-,Y_-)$, the field $\gamma$ blows up as
\begin{eqnarray}
\label{eq:inf_blowup}
\displaystyle \inf_{(x,y)\in \mathbb{T}^2} \gamma(x,y,t)  &\sim& 
\begin{cases}
-\frac{1}{2\lambda+1} \left(T_*-t\right)^{-1}, & -\frac{1}{2}<\lambda \leq 0\,,\\
-\frac{1}{\lambda+1} \left(T_*-t\right)^{-1}, & \quad 0<\lambda \,,
\end{cases}
\end{eqnarray}

\noindent (ii) At the pathline with initial position $(X(0),Y(0)) = (X_+,Y_+)$, the field $\gamma$ 
 blows up more mildly as:
\begin{eqnarray}
\label{eq:sup_blowup}
\displaystyle \sup_{(x,y)\in \mathbb{T}^2} \gamma(x,y,t) &\sim&
\begin{cases}
\frac{|\lambda|}{(\lambda+1)(2\,\lambda +1)} \left(T_*-t\right)^{-1}\,, &  -\frac{1}{2}<\lambda < 0\,,\\
\frac{1}{2\,\left|W_{-1}\left(-\pi \sqrt{S_* \Delta } \,(T_*-t)^{1/2}\right) \right|}  \left(T_*-t\right)^{-1}  \,, &   \lambda = 0\,,\\
 {\scriptstyle \frac{(\lambda+1)^{\frac{\lambda}{\lambda+1}}}{2\pi\sqrt{\Delta } \left\langle \left({\gamma_0(\cdot)} +|\gamma_-| \right)^{-\frac{1}{\lambda+1}}\right\rangle_{0}^{2\lambda+1}} } \left(T_*-t\right)^{-\frac{1}{\lambda+1}} , &  0<\lambda\,,
 \end{cases}
\end{eqnarray}
where $W_{-1}$ is the negative branch of the Lambert $W$ function.\\

\noindent (iii) At any pathline starting at neither $(X_-,Y_-)$ nor $(X_+,Y_+)$, the stretching blows up as in \eqref{eq:sup_blowup}. \\

\noindent (iv) The out-of-plane vorticity $\omega$ vanishes at the pathline with initial position $\mbf{X}(0) = (X_-,Y_-)$:
\begin{eqnarray}
\label{eq:vort_iv_-}
{\omega(\mbf{X}(t),t)} &\sim& 
{\omega_0(X_-,Y_-)}\times\begin{cases}
{\scriptstyle \left(2\pi |\lambda| (2\,\lambda+1)\sqrt{\Delta } \right)^{\frac{1}{2\,\lambda+1}} }\left(T_*-t\right)^{\frac{1}{2\,\lambda+1}}\,, &   -\frac{1}{2}<\lambda < 0\,,\\
 \frac{\pi \sqrt{\Delta }}{\left|W_{-1}\left(-\pi \sqrt{S_* \Delta } \,(T_*-t)^{1/2}\right) \right|} {\left(T_*-t\right)}\,, &   \lambda =0\,,\\
 \frac{{(\lambda+1)}^{\frac{1}{\lambda+1}}}{\left\langle \left({\gamma_0(\cdot)} + |\gamma_-| \right)^{-\frac{1}{\lambda+1}}\right\rangle_{0}}\,\left(T_*-t\right)^{\frac{1}{\lambda+1}}\,,& 0<\lambda\,,
\end{cases}
\end{eqnarray}
while, at any other pathline starting at $\mbf{X}_0 \neq (X_-, Y_-)$, it remains bounded if $0<\lambda$ and blows up if $-1/2<\lambda \leq 0$:
\begin{eqnarray}
\label{eq:vort_iv_+}
 \scriptstyle {\omega(\mbf{X}(t),t)}\sim\scriptstyle {\omega_0(\mbf{X}_0)}\times
\begin{cases}
\frac{\big((\lambda+1)(2\,\lambda+1) \left(|\gamma_-| + \gamma_0(\mbf{X}_0)\right)\big)^{\frac{1}{\lambda+1}} }{\left(2\pi |\lambda| (2\,\lambda+1)\sqrt{\Delta } \right)^{\frac{1}{2\,\lambda+1}}}  \left(T_*-t\right)^{-\frac{|\lambda|}{(\lambda+1)(2\,\lambda+1)}}\,, &  \scriptstyle -\frac{1}{2}<\lambda <0\,,\\
\frac{\left(|\gamma_-| + \gamma_0(\mbf{X}_0)\right)}{\pi \sqrt{\Delta }}  { \left|W_{-1}\left(-\pi \sqrt{S_* \Delta} \,(T_*-t)^{1/2}\right) \right|}\,, &  \scriptstyle \lambda =0\,.
\end{cases}
\end{eqnarray}
\end{theorem}

The proof is straightforward but long and is provided in the appendix.

\begin{remark}
\label{rem:blowup_and_zeroes}
Regarding this theorem, the blowup of the infimum of the stretching rate $\gamma$ is always accompanied by a collapse of the out-of-plane vorticity $\omega$ to zero, thus providing an example of theorem \ref{thm:conservation_of_zeroes} on the conservation of zeroes of the symmetry $\etab_1$ defined in \eqref{eq:symm_vort}. In contrast, the blowup of the supremum of the stretching rate $\gamma$ is not always accompanied by a blowup of $\omega$, as the case $0<\lambda$ shows.
\end{remark}

\subsection{Remarks on the spatial structure of the blowup}

The blowup of the infimum of the stretching rate $\gamma$ is stronger than the blowup of its supremum. In fact, for the exact 3D Euler solution case $\lambda=0$, the supremum blowup is marginal: the denominator is going to infinity logarithmically, making this blowup milder than the infimum blowup. In the case $0<\lambda$ the exponent in the supremum blowup is always smaller than in the infimum blowup. In the case $-1/2<\lambda<0$ the exponent is the same for both infimum and supremum, but the prefactor is always larger for the infimum blowup. 

Nevertheless, when considering the spatial structure of the blowup, the value of the stretching rate does not matter as much as the spatial profile of the blowup region.  Because $\gamma =- \nabla_{\mathrm{h}} \cdot\mbf{u}_{\mathrm{h}}$, considering the 2D flow on the $z=0$ plane an interesting pattern emerges: a blowup of the infimum of $\gamma$ corresponds to a positive-divergence flow, where the particles are repelled from the singular point. This makes such singularities ``thick'', and therefore amenable to numerical simulations, and when considering the third dimension, the flow is towards the plane, and as the out-of-plane vorticity tends to zero, the helicity is negligible near these singularities. See \cite{mulungye2016lambda_0, Ohkitani20003181, constantin2000euler} for numerical and analytical studies of the case $\lambda=0$, dominated by a blowup of the \mbox{infimum of $\gamma$.} In contrast, a blowup of the supremum of $\gamma$ corresponds to a negative-divergence flow, where the particles are attracted towards the singular point. These singularities are therefore ``thin'' and difficult to resolve numerically unless an adaptive spatial scheme is used. See \cite{mulungye2015lambda_m3/2} for a numerical and analytical study of the case $\lambda=-3/2$, dominated by a blowup of the supremum of $\gamma$. When considering the third dimension, the flow is a thin jet escaping the plane, and since both vorticity and velocity are large, the local helicity is singular there.

\section{Conclusion and discussion}
\label{sec:5}
In the first part of this paper the continuous symmetry approach is presented in a modern and practical way, and is applied with versatility to the task of finding the solution of 3D Euler fluid equations (and related models) along pathlines via the construction of constants of motion that depend on the velocity field, either locally or non-locally. In the second part of this paper, the solutions found are pursued further and conditions on finite-time blowup of the main fields of interest (out-of-plane vorticity and its stretching rate) are found and fully discussed. Our model serves as a platform to understand the complicated structure that might be encountered in nearly singular solutions of 3D Euler equations.

In future work we will extend the applications of the continuous symmetry approach to 3D Euler fluid models beyond the symmetry plane and to systems involving more interacting fields, such as magnetohydrodynamics. Moreover, we will combine the continuous symmetry approach with the current Cauchy invariants lore to generate a hierarchy of new Cauchy invariants, as theorem \ref{thm:construct_symm} is in fact a particular case of the result that the Lie derivative of any Lie-conserved tensor along a symmetry is a new Lie-conserved tensor, which is exact if the original one is.

\vskip6pt

\enlargethispage{20pt}






\ack{This paper is dedicated to the memory of Professor Charles R. Doering.}


\appendix

\section{Proof of theorem \ref{thm:conditions_-1/2<lambda}}

\begin{proof}[Proof of theorem \ref{thm:conditions_-1/2<lambda}]
The main ideas for this proof are: first, the terms in equations \eqref{eq:sol_gamma_S}--\eqref{eq:sol_omega_S} that can show singular behaviour are: the factor $1+(\lambda+1) \gamma_0(X_0,Y_0)S(t)$, which goes to zero at $t=T_*$ iff the pathline starts at $(X_-,Y_-)$; the factor $\dot S(t)$, which may go to zero at $t=T_*$; and the factor $\ddot S(t)$, which may go to infinity at $t=T_*$. Second, notice that equations \eqref{eq:sol_gamma_S}--\eqref{eq:sol_omega_S} verify $\gamma = \frac{d \ln |\omega|}{dt}$ along pathlines, so one can prove \emph{(iv)} first by inserting the asymptotic solutions into \eqref{eq:sol_omega_S} and then prove \emph{(i)--(iii)} by direct logarithmic differentiation (this is a formal proof, but is equivalent to a full proof in our case and far more efficient). Third, the supremum and infimum of $\gamma$ correspond to the pathlines starting at the supremum and infimum of $\gamma_0$ because the level sets of $\gamma$ maintain their order (see the short explanation after equation \eqref{eq:sol_omega_S}). We now turn  to prove \emph{(iv)}. We evaluate from \eqref{eq:sol_omega_S} the vorticity at the pathline starting at $(X_-,Y_-)$: $\omega(t) = \omega(0) \left(\frac{1-S/S_*}{{\dot S}^{1/2}}\right)^{\frac{1}{\lambda+1}}$, valid for $-1<\lambda$. When $-1/2 <\lambda\leq 0$, both numerator and denominator in this formula go to zero as $t\to T_*$. 
%
We can integrate equations \eqref{eq:asymptotics_S} and find:
\begin{equation}
\label{eq:asymptotics_S_of_t}
 S_*- S(t) \sim \begin{cases}
                     S_*^2\,\left[4\pi^2 \lambda^2 \Delta \right]^{\frac{\lambda+1}{2\,\lambda+1}}  \,\left(2\,\lambda+1\right)^{\frac{1}{2\,\lambda+1}}\,(T_*-t)^{\frac{1}{2\,\lambda+1}}\,, & \quad -1/2 < \lambda < 0\,, \\
		   S_* \exp \left[2\, W_{-1}\left(-\pi \sqrt{S_* \Delta } \,(T_*-t)^{1/2}\right)\right]\,, & \quad \lambda = 0\,,
                   \end{cases}
\end{equation}
where $W_{-1}$ is the negative branch of the Lambert $W$ function, solution of $W_{-1}(z)\mathrm{e}^{W_{-1}(z)} = z$ defined for small $z<0$ and such that $\lim_{z\to 0^-}W_{-1}(z) = -\infty$. Using these and equations \eqref{eq:asymptotics_S} it is direct to calculate the required ratio when $-1/2 < \lambda < 0$, thus proving the first line in \eqref{eq:vort_iv_-}. When $\lambda=0$ we have, defining $z=-\pi \sqrt{S_* \Delta } \,(T_*-t)^{1/2}$,
$$\frac{1-S/S_*}{{\dot S}^{1/2}} \sim \frac{\mathrm{e}^{2\, W_{-1}\left(z\right)}}{2\pi S_* \sqrt{\Delta} \left|\log(1-S(t)/S_*)\right|^{-1}} \sim  \frac{|W_{-1}\left(z\right)|\,\mathrm{e}^{2\,W_{-1}\left(z\right)}}{\pi S_* \sqrt{\Delta}} =   \frac{z^{2}}{\pi S_* \sqrt{\Delta} \,|W_{-1}\left(z\right)|},$$
which proves the second line in \eqref{eq:vort_iv_-}. When $0 <\lambda$ the denominator $\dot S^{1/2}$ is finite at $T_*$, as shown in the proof of theorem \ref{thm:blowup_lambda<-1}. The numerator can be approximated to lowest order and we get
$$\frac{1-S/S_*}{{\dot S}^{1/2}} \sim \frac{\dot S(T_*)}{S_* \dot S(T_*)^{1/2}}(T_*-t) = \frac{\dot S(T_*)^{1/2}}{S_*}(T_*-t),$$
and from \eqref{eq:formula_S} we get \mbox{$\dot S(T_*) = {\gamma_-^{-2}}\left\langle \left({\gamma_0(\cdot)} + {|\gamma_-|} \right)^{-\frac{1}{\lambda+1}}\right\rangle_{0}^{-2(\lambda+1)}$.} Recalling $S_*^{-1}=(\lambda+1) |\gamma_-|$ and simplifying leads to the third line in  \eqref{eq:vort_iv_-}.  Now, we consider the vorticity at a pathline starting at $\mbf{X}_0 \neq (X_-,Y_-)$: from \eqref{eq:sol_omega_S} we get $\omega(t) = \omega(0) \left(\frac{1+(\lambda+1) \gamma_0(\mbf{X}_0) S}{{\dot S}^{1/2}}\right)^{\frac{1}{\lambda+1}}$, and now the numerator does not go to zero as $t \to T_*$. The ratio will blowup if $\dot S$ tends to zero as $t \to T_*$. Thus, when $0 < \lambda$ the vorticity remains finite as $\dot S(T_*)$ is finite. But when $-1/2<\lambda\leq 0$ we have blowup:
$$\frac{1+(\lambda+1) \gamma_0(\mbf{X}_0) S}{{\dot S}^{1/2}} \sim \frac{1+(\lambda+1) \gamma_0(\mbf{X}_0) S_*}{{\dot S}^{1/2}} = \frac{|\gamma_-| + \gamma_0(\mbf{X}_0)}{|\gamma_-| {\dot S}^{1/2}},$$
which after using Eqs.~\eqref{eq:asymptotics_S} and \eqref{eq:asymptotics_S_of_t} gives directly the first line in \eqref{eq:vort_iv_+} when \mbox{$-1/2<\lambda<0$.} When $\lambda=0$ we get, in terms of the $z$ variable defined earlier,
$$\frac{|\gamma_-| + \gamma_0(\mbf{X}_0)}{|\gamma_-| {\dot S}^{1/2}} \sim \frac{|\gamma_-| + \gamma_0(\mbf{X}_0)}{2\pi S_* \sqrt{\Delta} |\gamma_-|  \left|\log(1-S(t)/S_*)\right|^{-1}} \sim \frac{(|\gamma_-| + \gamma_0(\mbf{X}_0)) |W_{-1}\left(z\right)|}{\pi \sqrt{\Delta}},$$
thus proving the second line in \eqref{eq:vort_iv_+}. This concludes the proof of \emph{(iv)}.

We now turn to the proof of \emph{(i)}, which concerns the pathline starting at $(X_-,Y_-)$. From logarithmic differentiation of \eqref{eq:vort_iv_-}, we get
$$\gamma(\mbf{X}(t),t) = \frac{d}{dt}[\ln\omega(\mbf{X}(t),t)] \sim \begin{cases}
\frac{d}{dt} \ln \left(T_*-t\right)^{\frac{1}{2\,\lambda+1}}\,, &   -\frac{1}{2}<\lambda < 0\,,\\
 \scriptstyle - \frac{d}{dt} \ln{\left|W_{-1}\left(-\pi \sqrt{S_* \Delta } \,(T_*-t)^{1/2}\right) \right|} + \frac{d}{dt} \ln{\left(T_*-t\right)}\,, &   \lambda =0\,,\\
  \frac{d}{dt} \ln\left(T_*-t\right)^{\frac{1}{\lambda+1}}\,,& 0<\lambda\,,
\end{cases}$$
which leads to \eqref{eq:inf_blowup} as the last term in the $\lambda=0$ case above dominates over the first term.

We now turn to the proof of \emph{(ii)--(iii)}, corresponding to pathlines starting at $\mbf{X}_0 \neq (X_-,Y_-)$. As explained earlier, the asymptotic behaviour of $\gamma$ along these pathlines is similar because, from equation \eqref{eq:sol_gamma_S} the first term is bounded while the second term, being proportional to $\ddot S / \dot S$ and thus independent of the initial point (thus ensuring \emph{(iii)} follows from \emph{(ii)}), blows up because $\ddot S$ blows up (and moreover, when $-1/2<\lambda\leq 0$, because $\dot S$ tends to zero). As the level sets of $\gamma$ keep their ordering, the supremum of $\gamma$ is always attained at the pathline starting at $(X_+,Y_+)$. When $-1/2<\lambda\leq 0$, logarithmic differentiation of \eqref{eq:vort_iv_+} gives
$$\gamma(\mbf{X}(t),t) = \frac{d}{dt}[\ln\omega(\mbf{X}(t),t)] \sim \begin{cases}
\frac{d}{dt} \ln \left(T_*-t\right)^{-\frac{|\lambda|}{(\lambda+1)(2\,\lambda+1)}}\,, &   -\frac{1}{2}<\lambda < 0\,,\\
 \scriptstyle \frac{d}{dt} \ln{\left|W_{-1}\left(-\pi \sqrt{S_* \Delta } \,(T_*-t)^{1/2}\right) \right|} \,, &   \lambda =0\,,
\end{cases}$$
proving directly the first line in  \eqref{eq:sup_blowup}, namely when $-1/2 < \lambda <0$. The $\lambda=0$ case gives, using the identity $\mathrm{d} \ln |W_{-1}(z)| = (1+W_{-1}(z))^{-1} \mathrm{d}\ln |z|$,  
$$\lambda=0: \qquad \gamma(\mbf{X}(t),t) \sim \frac{1}{1+W_{-1}(z)} \frac{d}{dt}\ln |z| \sim \frac{-\frac{1}{2}(T_*-t)^{-1}}{W_{-1}(z)},$$
giving the second line in \eqref{eq:sup_blowup}.  Finally, the case $0<\lambda$ requires the calculation of $\ddot S$ asymptotically, which blows up, while $\dot S(T_*)$ is finite, so $\gamma(\mbf{X}(t),t) \sim -\frac{1}{2(\lambda+1)} \frac{\ddot S(t)}{\dot S(T_*)}$. Now, from \eqref{eq:ddot_S} for $0<\lambda$ evaluated near the singularity we get, following a similar calculation as the one leading to \eqref{eq:dot_S_generic}, 
$$\ddot S \sim \frac{(\lambda+1)\gamma_-}{\pi \sqrt{\Delta}} (\dot S(T_*))^{2+\frac{1}{2(\lambda+1)}}\int_0^\epsilon (1-S(t)/S_* + (\lambda+1) S_* r^2/2)^{-\frac{\lambda+2}{\lambda+1}} r \mathrm{d}r\,,$$
where $\epsilon$ is small but fixed. The integral is elementary and the term at $r=0$ is divergent as $t\to T_*$ so it dominates, and we get
$\ddot S \sim \frac{(\lambda+1)\gamma_-}{\pi S_* \sqrt{\Delta}} (\dot S(T_*))^{2+\frac{1}{2(\lambda+1)}} \left(1-S(t)/S_*\right)^{-\frac{1}{\lambda+1}}.$ Writing to lowest order $1-S(t)/S_* \sim \frac{\dot S(T_*)}{S_*} (T_*-t)$ and
recalling that \mbox{$\dot S(T_*) = {\gamma_-^{-2}}\left\langle \left({\gamma_0(\cdot)} + {|\gamma_-|} \right)^{-\frac{1}{\lambda+1}}\right\rangle_{0}^{-2(\lambda+1)}$} and $S_*^{-1}=(\lambda+1)|\gamma_-|$ we get the third line in \eqref{eq:sup_blowup}, thus ending the proof.
\end{proof}

\bibliographystyle{rsta}
\bibliography{bibli_Cont_Symm}

\end{document}